\documentclass[a4paper,onecolumn,11pt,noarxiv]{quantumarticle}

\pdfoutput=1
\usepackage{amsmath}
\usepackage{amsfonts}
\usepackage{amsthm}
\usepackage{amssymb}
\usepackage{graphicx}
\usepackage{enumerate}
\usepackage{color,xcolor}
\usepackage[T1]{fontenc}
\usepackage{braket}
\usepackage{todonotes}
\usepackage{soul}
\usepackage{subfigure}
\usepackage{eucal}
\usepackage{ulem}
\usepackage[numbers]{natbib}
\graphicspath{{figures/}}

\usepackage[bookmarks=false,colorlinks,citecolor=blue,
linkcolor=blue,anchorcolor=blue,urlcolor=blue
]{hyperref}

\newtheorem{theorem}{Theorem}
\newtheorem{lemma}{Lemma}
\begin{document}

%\title{A universal framework for detecting genuine multipartite entanglement}
\title{A generic framework for genuine multipartite entanglement detection}

\author{Xin-Yu Xu}
\author{Qing Zhou}
\author{Shuai Zhao}
\author{Shu-Ming Hu}
\affiliation{Hefei National Research Center for Physical Sciences at the Microscale and School of Physical Sciences, University of Science and Technology of China, Hefei 230026, China}
\affiliation{CAS Center for Excellence in Quantum Information and Quantum Physics, University of Science and Technology of China, Hefei 230026, China}
\author{Li Li}
\email{eidos@ustc.edu.cn}
\author{Nai-Le Liu}
\email{nlliu@ustc.edu.cn}
\author{Kai Chen}
\email{kaichen@ustc.edu.cn}
\affiliation{Hefei National Research Center for Physical Sciences at the Microscale and School of Physical Sciences, University of Science and Technology of China, Hefei 230026, China}
\affiliation{CAS Center for Excellence in Quantum Information and Quantum Physics, University of Science and Technology of China, Hefei 230026, China}
\affiliation{Hefei National Laboratory, University of Science and Technology of China, Hefei 230088, China}

\maketitle

\begin{abstract}
	%Certification of multipartite entanglement is of central interest for both fundamental research and essential application in quantum information tasks. Entanglement witness is the most widely used and efficient tool for experimental detection. When the number of particles grows rapidly like in a quantum network, however, accurate and robust entanglement witness often poses a significant challenge. Here we push forward an exquisite design of robust entanglement witness with a universal and operational framework. The systematic construction allows to substantiate genuine multipartite entanglement for a variety of states that arise naturally in practical situations, and to outperform dramatically many existing and standard methods. Furthermore, it achieves excellent noise tolerance for many typically entangled states that are demanding in extensive scenarios of quantum applications.
	Design of detection strategies for multipartite entanglement stands as a central importance on our understanding of fundamental quantum mechanics and has had substantial impact on quantum information applications. However, accurate and robust detection approaches are severely hindered, particularly when the number of nodes grows rapidly like in a quantum network. Here we present an exquisite procedure that generates novel entanglement witness for arbitrary targeted state via a generic and operational framework. The framework enjoys a systematic and high-efficient character and allows to substantiate genuine multipartite entanglement for a variety of states that arise naturally in practical situations, and to dramatically outperform currently standard methods. With excellent noise tolerance, our framework should be broadly applicable to witness genuine multipartite entanglement in various practically scenarios, and to facilitate making the best use of entangled resources in the emerging area of quantum network.
\end{abstract}

\section{Introduction} \label{sec:introduction}

As a unique property in quantum theory, entanglement \cite{RevModPhys.81.865} is recognized as a kind of quantum resource \cite{RevModPhys.91.025001} and plays a central role in numerous quantum computing and quantum communication tasks \cite{bennett2000quantum,PhysRevLett.70.1895,RevModPhys.81.1301,feynman1982simulating,deutsch1985quantum}. The ability to generate an increasing number of entangled particles is an essential benchmark for quantum information processing. In past decades, considerable efforts have been made to prepare larger and more complex entangled states in various platforms \cite{Luo620,arute2019quantum,PhysRevLett.105.210504,yao2012,doi:10.1126/science.abg7812,yokoyama2013ultra,PhysRevLett.112.155304}, which experimental systems are currently evolving from several qubits to noisy intermediate scale quantum system (NISQ) \cite{Preskill2018quantumcomputingin}. 

The developments of quantum technologies raise immediately important questions regarding characterization of quantum entanglement of underlying systems. In bipartite systems, various theoretical works have been contributed, such as separability criterions \cite{PhysRevLett.77.1413,HORODECKI1997333,chen205017quantum,rudolph2005further} and entanglement measures \cite{PhysRevLett.95.040504,PhysRevLett.95.210501,plenio2014introduction}, which provide standard tools for characterizing bipartite entanglement. For good reviews, please refer to Refs.\ \cite{RevModPhys.81.865,GUHNE20091,friis2019entanglement}. When it comes to multipartite systems, the problem is much more complicated. The entanglement structure becomes much richer for multipartite systems \cite{PhysRevLett.108.110501,zhou2019detecting}, since the number of possible divisions grows exponentially with the system size \cite{RevModPhys.81.865}. This leads to many types of multipartite entanglement, ranging from non-fully-separable to genuine multipartite entanglement (GME). In the following, we focus on the detection of genuine multipartite entanglement, which is an essential task for multipartite quantum communication and quantum computing tasks. For the detection of GME, many standard tools in the bipartite case, such as separability criterions, become infeasible since they only detect entanglement between two partitions. Meanwhile, a tomographic reconstruction of quantum state required in these methods becomes time-consuming and computationally difficult in the multipartite case.

For genuine multipartite entanglement detection, entanglement witness (EW) \cite{HORODECKI19961,terhal2000bell,lewenstein2000optimization,hyllus2005relations,lewenstein2001characterization} provides an elegant solution both theoretically and experimentally without need of having full tomographic knowledge about the state. Moreover, it is also known that witness operator can also be used to estimate entanglement measures \cite{PhysRevLett.98.110502}. On account of simplicity and efficiency of entanglement witness, it has been widely used for experimental certification of GME in many platforms, such as trapped ions \cite{PhysRevX.8.021012,PhysRevLett.106.130506}, photonic qubits \cite{PhysRevLett.95.210502,gao2010experimental,PhysRevX.8.021072,PhysRevLett.124.160503}, and superconducting qubits\cite{PhysRevLett.122.110501}. Most available GME witnesses are tailored towards some specific states, for instance, the Greenberger-Horne-Zeilinger (GHZ) states \cite{Greenberger1989}, W-states \cite{PhysRevA.62.062314}, graph states \cite{PhysRevA.69.062311,hein2006entanglement}, and so on. Despite few general methods for the construction of GME witness have been proposed \cite{PhysRevLett.92.087902,PhysRevLett.106.190502,PhysRevLett.113.100501,PhysRevLett.111.110503}, their performance is very limited, especially as the size of system grows. One major drawback is the limited scope of noise resistance. For example, the fidelity-based method \cite{PhysRevLett.92.087902} is a canonical witness construction and widely used nowadays. Its noise tolerance decreases dramatically as the system size increases.  In realistic NISQ systems, however, the noise always inevitably grows with the system size. In fact, it has been shown that the fidelity witnesses fail to detect a large amount of mixed entangled states \cite{PhysRevLett.124.200502}. To find more robust GME witnesses, numerical methods have been introduced \cite{PhysRevLett.106.190502},which, however, suffer from expensive computational costs as the system size grows. Hence, although it is known that for any entangled state there exists some EW to detect it \cite{HORODECKI19961}, how to construct a desirable EW to recognize a GME state is still a formidable challenge.

In this work, we propose a generic framework to design robust GME witnesses by analytical and systematic construction. We start by introducing an exquisite method for GME witness with a novel lifting from any set of bipartite EWs. This establishes the link between the standard tools developed in the bipartite case and the GME witness construction. We then provide a well-designed class of optimal bipartite EWs that allows the design of robust GME witnesses for arbitrary pure GME states with our method. The performance of this framework on many typical classes of GME states is further evaluated in terms of white noise tolerance. It can be shown that the framework outperforms the most widely used fidelity-based method with certainty, and outperforms much better than the best known EWs in many cases. Finally, benefiting from the high robustness of the resulting witnesses, we also demonstrate further applications of the framework, such as to provide a tighter lower bounds on the genuine multipartite entanglement measures and detecting unfaithful GME states \cite{PhysRevLett.124.200502}.

\section{Results}

\subsection{Preliminaries}

To start with, we first give the precise definition of biseparable, genuine multipartite entanglement and entanglement witness. A pure state is called \textit{biseparable} if it can be written as a tensor product of two state vectors, i.e., $|\psi_A\rangle\otimes|\psi_{\bar{A}}\rangle$. Then a mixed state is called biseparable if it can be decomposed into a mixture of pure biseparable states, formally,
\begin{equation}
	\rho_{bs}=\sum_{A|\bar{A},i}p_{A|\bar{A},i} |\psi_A^i\rangle\langle\psi_A^i|\otimes|\psi_{\bar{A}}^i\rangle\langle\psi_{\bar{A}}^i|,
\end{equation}
where the summation can be performed over different bipartitions $A|\bar{A}$ of the whole system. A state that is not biseparable is referred to as genuine multipartite entangled. To detect the GME states, the most widely used method is to find an observable $\mathcal{W}_{GME}$ that is nonnegative for all separable states and has negative expectation value on at least one GME state. Then for some multipartite quantum state $\rho$, the fact $Tr(\mathcal{W}_{GME}\rho) < 0$ will reveal the existence of genuine multipartite entanglement, and the $\mathcal{W}_{GME}$ is called a \textit{GME witness}. Moreover, given two EWs $\mathcal{W}_1$ and $\mathcal{W}_2$, if there exists $\lambda > 0$ such that $\mathcal{W}_1 - \lambda \mathcal{W}_2$ is positive semidefinite, i.e., $\mathcal{W}_1 \succeq \lambda \mathcal{W}_2$, one says that $\mathcal{W}_2$ is \textit{finer} than $\mathcal{W}_1$ \cite{lewenstein2000optimization}. The finer witness operator $\mathcal{W}_2$ detects more entangled states than $\mathcal{W}_1$. An EW is \textit{optimal} if no finer EW exists. 

\subsection{Design GME witness from a complete set of bipartite EWs}
Due to its non-negativity over all biseparable states, a GME witness $\mathcal{W}_{GME}$ also serves as bipartite EW with respect to each possible bipartition of the whole system. In other words, there exists a complete set of bipartite EWs $\{\mathcal{W}_{A|\bar{A}}\}$ satisfying $\mathcal{W}_{GME} \succeq \mathcal{W}_{A|\bar{A}}$ for each bipartition $A|\bar{A}$. This fact, from the opposite point of view, indicates that the GME witness $\mathcal{W}_{GME}$ is designed based on the set $\{\mathcal{W}_{A|\bar{A}}\}$ according to the constraint $\mathcal{W}_{GME} \succeq \mathcal{W}_{A|\bar{A}}$. This naturally provides a general framework for constructing GME witnesses from a complete set of bipartite EWs. Remarkably, the set $\{\mathcal{W}_{A|\bar{A}}\}$ itself cannot be directly used to detect GME states, as there exist biseparable states that are entangled with respect to every possible bipartition \cite{GUHNE20091}. While there are two crucial issues with such a framework. The first one is how to find the operator satisfying $\mathcal{W}_{GME} \succeq \mathcal{W}_{A|\bar{A}}$, and the second one is to decide which set of bipartite EWs should be used. Optimal solutions to these two problems is hard in general, and there have been only a few previous related studies on these issues. In Ref.\ \cite{PhysRevLett.113.100501}, an alternative solution was proposed to establish a connection between positive maps and multipartite EWs, where EWs detecting multipartite bound entangled state have been obtained. While in the following, we present a novel alternatively solution which is capable of constructing robust GME witnesses. 

\subsection{An operational framework for constructing robust GME witness}

Any mixed GME state contains at least one pure GME state as a component, while the remaining components can be treated as noises. In order to detect mixed GME states with linear EW, it is natural to employ a witness operator for the pure GME component that is sufficiently robust to noise from the other components. In fact, the set of all optimal GME witnesses for all pure GME states will be sufficient to detect all GME states. However, finding all optimal GME witnesses is naturally a formidable task. Therefore, to advance a solution to this problem, we propose an operational framework to construct a class of robust GME witnesses for all pure GME states. 

To address the problem of lifting any given set of bipartite EWs to multipartite, one can accomplish it in two steps: (1) For the first step, each bipartite EW $\mathcal{W}_{A|\bar{A}}$ is decomposed into some projectors. Note that the entanglement witness is designed for some pure entangled state $|\psi\rangle$. Hence we extract a term $-|\psi\rangle\langle\psi|$ before the decomposition. That is, the bipartite EWs are rewritten as $\mathcal{W}_{A|\bar{A}} = \mathcal{O}_{A|\bar{A}} - |\psi\rangle\langle\psi|$, and a spectral decomposition of $\mathcal{O}_{A|\bar{A}} = \mathcal{W}_{A|\bar{A}}+|\psi\rangle\langle\psi|$ is performed
\begin{equation}\label{bEW}
	\mathcal{O}_{A|\bar{A}}=\sum_{|\vec{v}_{i,A|\bar{A}}\rangle \in \mathcal{S}_{A|\bar{A}}}
	c_{i,A|\bar{A}} |\vec{v}_{i,A|\bar{A}}\rangle\langle \vec{v}_{i,A|\bar{A}}|,
\end{equation}
with $\mathcal{S}_{A|\bar{A}}$ being the set of eigenvectors and $c_{i,A|\bar{A}}$ being the corresponding eigenvalues. All these eigenvectors are collected into a set $\mathcal{S} = \cup_{A|\bar{A}} \mathcal{S}_{A|\bar{A}}$. (2). For the second step, the obtained set $\mathcal{S}$ is divided into $m$ subsets $\{\mathcal{S}_k\}_{k=1}^m$, such that the vectors from different subsets are orthogonal with each other. Denote $\tilde{I}_k$ as the identity operator on the subspace $V_k$ spanned by the state vectors from subset $\mathcal{S}_k$, and $c_k=max_{|\vec{v}_{i,A|\bar{A}}\rangle \in \mathcal{S}_k} c_{i,A|\bar{A}} $ as the maximal coefficient attached to the state vectors in $\mathcal{S}_k$. With the above preparation and notation, we proceed to the following Theorem:
\begin{theorem}
Given any pure GME state $|\psi\rangle$ and a set of bipartite EWs $\{\mathcal{W}_{A|\bar{A}}\}$ detecting $|\psi\rangle$ for all possible $A|\bar{A}$, the following operator $\mathcal{\hat{W}}$
\begin{equation}
	\mathcal{\hat{W}}=\sum_{k=1}^m c_k \tilde{I}_k -|\psi\rangle\langle\psi|,
\end{equation}
is nonnegative over all biseparable states, where the $c_k$ and $\tilde{I}_k$ have been defined above.
\end{theorem}
\begin{proof}
To prove the statement, it suffices to observe
\begin{equation}
	\begin{aligned}
		\mathcal{\hat{W}}-\mathcal{W}_{A|\bar{A}}=&\sum_{k=1}^m c_k \tilde{I}_k -\mathcal{O}_{A|\bar{A}} \\
		=&\sum_{k=1}^m \left(c_k \tilde{I}_k-\sum_{|v_{i,A|\bar{A}}\rangle \in \mathcal{S}_k \cap \mathcal{S}_{A|\bar{A}}} c_{i,A|\bar{A}}
		|\vec{v}_{i,A|\bar{A}}\rangle\langle \vec{v}_{i,A|\bar{A}}|\right) \\
		\ge& \sum_{k=1}^m c_k\left(\tilde{I}_k-\sum_{|v_{i,A|\bar{A}}\rangle \in \mathcal{S}_k \cap \mathcal{S}_{A|\bar{A}}}
		|\vec{v}_{i,A|\bar{A}}\rangle\langle \vec{v}_{i,A|\bar{A}}|\right) \\
		\ge& 0,
	\end{aligned}
\end{equation}
where the inequalities can be derived directly from the definitions of $c_k$ and $\tilde{I}_k$.
\end{proof}

The above construction can be interpreted geometrically. That is, noise from different subspaces has different degrees of influence on the entanglement properties of the target state. The influence is characterized by the coefficients $c_k$, and a small $c_k$ indicates that noise from this subspace hardly affects the entanglement property of target state. Therefore, Theorem 1 can be seen as robust GME witness construction with the help of some prior knowledge of the target state, which comes from the set of bipartite EWs $\{\mathcal{W}_{A|\bar{A}}\}$.

Remarkably, Theorem 1 itself cannot be used as an operational framework for GME witness construction, since the resulting operators can be positive semidefinite and fail to detect any GME state. In fact, one can hardly expect a nontrivial result when the set of bipartite EWs $\mathcal{W}_{A|\bar{A}}$ are chosen randomly. Fortunately, standard tools exist for constructing bipartite EWs based on positive maps. In the following, in order to obtain an operational and generic framework for GME witness construction, we provide a promising choice on the set of bipartite EWs, which are designed for the target states based on partial transposition. 

Under any given bipartition $A|\bar{A}$, the target state $|\psi\rangle$ can be written in a Schmidt decomposition form $|\psi\rangle=\sum_{i=0}^{r_A-1} \sqrt{\lambda_{i,~A|\bar{A}}}|i_Ai_{\bar{A}}\rangle$, with $r_A$ being the corresponding Schmidt rank. Note that here the local dimension of the Hilbert space need not be fixed. Then we introduce a class of bipartite EWs $\mathcal{W}_{o,A|\bar{A}}$ in order to use them in the construction of GME witness.
\begin{equation}\label{eq:obEW}
	\mathcal{W}_{o,A|\bar{A}}= \sum_{i,j=0}^{r_A-1} \sqrt{\lambda_{i,~A|\bar{A}}\lambda_{j,~A|\bar{A}}}
	|i_Aj_{\bar{A}}\rangle\langle i_Aj_{\bar{A}}|-|\psi\rangle\langle\psi|.
\end{equation}
The choice of $\mathcal{W}_{o,A|\bar{A}}$ is mainly based on two considerations. Firstly, $\mathcal{W}_{o,A|\bar{A}}+|\psi\rangle\langle\psi|$ naturally takes the decomposition form in the Eq.\ (\ref{bEW}). Secondly, the above $\mathcal{W}_{o,A|\bar{A}}$ are a class of optimal bipartite EWs. For a detailed illustration and discussion on the $\mathcal{W}_{o,A|\bar{A}}$, please refer to Appendix.\ \ref{sec:appendix bipartite EW}.

These bipartite EWs, together with Theorem 1, promise a generic framework to construct GME witnesses with certainty. The explicit procedure is as follows:
\begin{enumerate}[(1).]
    \item Firstly, find the set $\mathcal{S}$. For each bipartition $M|\bar{M}$, calculate the Schmidt decomposition of $|\psi\rangle$ with respect to $M|\bar{M}$, 
    \begin{equation}
        |\psi\rangle=\sum_{i=0}^{r_{M|\bar{M}}-1} \lambda_{i,M|\bar{M}}|\varphi_{i,M}\rangle|\varphi_{i,\bar{M}}\rangle,
    \end{equation}
    with $r_{M|\bar{M}}$ being the Schmidt rank under this bipartition. A total of $r^2_{M|\bar{M}}$ vectors will be added to the set $\mathcal{S}$, and each of them has a corresponding coefficient. This is denoted by 
    \begin{equation}
        \{(\sqrt{\lambda_{i,M|\bar{M}}\lambda_{j,M|\bar{M}}},~|\varphi_{i,M}\rangle|\varphi_{j,\bar{M}}\rangle)\}_{i,j=0}^{r_{M|\bar{M}}-1}. 
    \end{equation}
    After traversing all possible bipartitions, one will end up with a set of vectors $\mathcal{S}$ as well as their corresponding coefficients, that is, $\{(c_k,~|\psi_k\rangle)\}_{k=1}^{|\mathcal{S}|}$. 

    \item Secondly, find the finest division of $\mathcal{S}$ such that vectors from different subsets are orthogonal with each other. This can be achieved with the following steps: 
    \begin{enumerate}[(i)]
    	\item
    	Put the first element $|\psi_1\rangle$ of $\mathcal{S}$ into an empty subset $\mathcal{S}_1$. 
    	\item
   	 	For every other vector in $\mathcal{S}-\mathcal{S}_1$, if it is not orthogonal with all vectors in the set $\mathcal{S}_1$, it is added into $\mathcal{S}_1$. Repeat this step until no new vector can be added to $\mathcal{S}_1$.
    	\item
    	For the rest vectors in $\mathcal{S}-\mathcal{S}_1$, repeat the above two steps to obtain $\mathcal{S}-\mathcal{S}_1-\mathcal{S}_2$, $\mathcal{S}-\mathcal{S}_1-\mathcal{S}_2-\mathcal{S}_3$, $\cdots$, until one has classified all the elements of $\mathcal{S}$. 
    	\item 
    	One obtain a division $\mathcal{S}=\sum_{k=1}^{m}\mathcal{S}_k$. 
    \end{enumerate}

    \item Thirdly, calculate the subspace spanned by the vectors in subset $\mathcal{S}_k$. By performing Schmidt orthogonalization of the vectors in $\mathcal{S}_k$, one can derive the subspace spanned by these vectors and obtain the identity operator $\tilde{I}_k$ on this subspace. 
    
    \item Finally, for each subset $\mathcal{S}_k$, find the maximal coefficients $c_k$ attached to the vectors in it, and construct a GME witness using Theorem 1.
\end{enumerate}

There are two remarks to note about this method. Firstly, the resulting witness from the above procedure is always finer than the commonly used GME fidelity witness $\mathcal{W}_F =\lambda I-|\psi\rangle\langle\psi| $ for $|\psi\rangle$, with $\lambda=\max_{A|\bar{A}} \lambda_{0,A|\bar{A}}$ (Here it is assumed that the Schmidt coefficients $\lambda_{i,A|\bar{A}}$ are in decreasing order). To illustrate this, note that if the bipartite EWs are chosen as the bipartite fidelity witness $\mathcal{W}_{F,A|\bar{A}} = \lambda_{0,A|\bar{A}} I-|\psi\rangle\langle\psi|$, by applying Theorem 1, the obtained operator is nothing but the $\mathcal{W}_F$. Whereas by checking $\mathcal{W}_{F,A|\bar{A}}-\mathcal{W}_{o,A|\bar{A}} \succeq 0$, it is straightforward to verify that the above $\mathcal{W}_{o,A|\bar{A}}$ is finer than the bipartite fidelity witness $\mathcal{W}_{F,A|\bar{A}}$. Therefore, when Theorem 1 is applied to the set of $\mathcal{W}_{o,A|\bar{A}}$, the resulting GME witness strictly outperforms the corresponding fidelity witness $\mathcal{W}_{F}$. Secondly, one starts from a complete set of bipartite EWs in the above construction, leading to EWs that detect genuine multipartite entanglement. While if one starts from a smaller set of bipartite EWs, the method allows also for flexible applications in verifying other kinds of multipartite entanglement, e.g., characterizing the entanglement depth.

\section{Examples}

To help a better understanding as well as quantitatively investigating the robustness of the framework, we proceed to some explicit examples, where the white noise tolerance is employed as a figure of merit to evaluate its performance in practice. The white noise tolerance of some witness operator $\mathcal{W}$ for $|\psi\rangle$ is the critical value of $p$ such that the mixed state $pI/d+(1-p)|\psi\rangle\langle\psi|$ is not detected by $\mathcal{W}$.

\subsection{$W$-state} 
To investigate the asymptotic behavior of this framework with an increasing system size, we start with the $n$-qubit $W$-state $|W_n\rangle=(|00\cdots01\rangle$$+|00\cdots10\rangle+$$\cdots+|10\cdots00\rangle)/\sqrt{n}$, which is widely used in quantum information processing tasks. For the $W$-state, we find the GME witness (see Appendix.\ \ref{sec:appendix w state} for a proof.)
\begin{equation}
	\mathcal{W}_{|W_n\rangle}=\frac{n-1}{n}\mathcal{P}^n_1+\frac{\sqrt{\lfloor n/2\rfloor(n-\lfloor n/2\rfloor)}}{n}
	(\mathcal{P}^n_0 +\mathcal{P}^n_2)-|W_n\rangle\langle W_n|,
\end{equation}
with $\mathcal{P}^n_i=\sum_m\pi_m(|0\rangle^{\otimes n-i}|1\rangle^{\otimes i})\pi_m(\langle 0|^{\otimes n-i}\langle 1|^{\otimes i})$, where the summation $m$ is over all possible permutation of $|0\rangle^{\otimes n-i}|1\rangle^{\otimes i}$. The $\mathcal{W}_{|W_n\rangle}$ recovers a class of EWs presented in Ref.\ \cite{Bergmann_2013}, which are the most powerful ones for the $W$-state presently. Its white noise tolerance also tends to 1 for an increasing number of qubits. While for the fidelity witness, its white noise tolerance is $1/(n(1-1/2^n))$, tending to $1/n$ for large $n$.

\subsection{Graph state}
Graph states are a class of genuine multipartite entangled states that are of great importance for measurement-based quantum computation \cite{briegel2009measurement} and quantum error correction \cite{PhysRevA.65.012308}, etc. In Refs.\ \cite{PhysRevLett.106.190502,PhysRevA.84.032310}, the authors have developed powerful entanglement witnesses for graph states. While our framework suggests that there is still much room for improvement in the robustness of these existing results.

More specifically, we focus on a typical class of graph state---the $n$-qubit ($n\ge4$) linear cluster states $|Cl_n\rangle$ in this example. The $|Cl_n\rangle$ can be expressed by a set of stabilizers $\{g_i\}_{i=1}^n$, with $g_i=Z_{i-1}X_iZ_{i+1}$ ($2\le i \le n-1$), $g_1=X_1Z_2$ and $g_n=Z_{n-1}X_n$ respectively, where the $X$ and $Z$ are Pauli operators. All the common eigenstates of these stabilizers introduce a complete basis, i.e., the graph state basis. This basis can be denoted by $|\vec{a}\rangle_{Cl_n}$, with $\vec{a}=a_1a_2\cdots a_n\in \{0,1\}^n$, such that $g_i|\vec{a}\rangle_{Cl_n}=(-1)^{a_i}|\vec{a}\rangle_{Cl_n}$ for $i=1,\cdots,n$. Specially, the $|Cl_n\rangle$ corresponds to $|00\cdots0\rangle_{Cl_n}$. When applied to the linear cluster state, our framework results in a GME witness which is diagonal under the graph state basis, (For the explicit construction process, we refer to Appendix.\ \ref{sec:appendix graph state}.)
\begin{equation}
	\mathcal{W}_{Cl_n}=\sum_{k=1}^{\lceil n/3 \rceil}
	\sum_{\vec{a}\in V_k}\frac{1}{2^k-1}|\vec{a}\rangle_{Cl_n}\langle\vec{a}|-|Cl_n\rangle\langle Cl_n|.
\end{equation}
Here a vector $\vec{a}$ belongs to $V_k$ if there exist at most $k$ for the number of `$1$'s in $\vec{a}$, such that their distance with each other are larger than $2$ at the same time (for instance, $1101100$ belongs to $V_2$ while $1001011$ belongs to $V_3$).
\begin{figure*}[h]
	\centering
	\includegraphics[width=10.67cm,height=8cm]{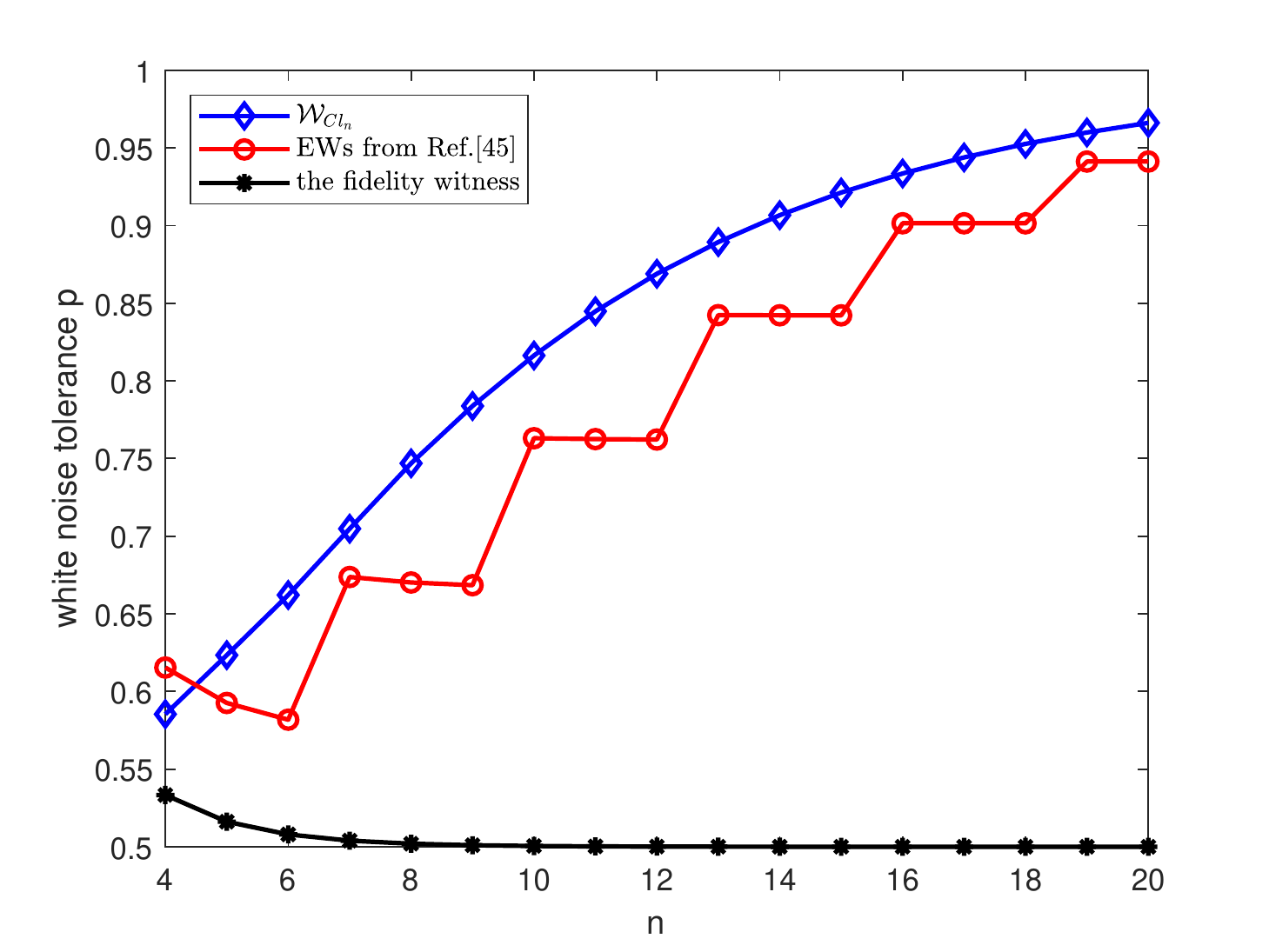}
	\caption{
	In this figure we illustrate the performance of $\mathcal{W}_{Cl_n}$ by showing its white noise tolerance for $n$-qubit cluster state up to $20$ qubits. It is plotted with a blue line with diamonds. As a comparison, the red circle line is the best known results introduced in the Ref.\ \cite{PhysRevLett.106.190502}. The result in this paper outperforms this existing EW since the number of qubit is larger than $5$. Both of their white noise tolerance tend to $1$. Remarkably, for the widely used fidelity witness of $|Cl_n\rangle$, its white noise tolerance is given by the black line with filled circles tends to $1/2$ with an increasing qubit number.
	}	
	\label{fig:fig1}
\end{figure*}
Its white noise tolerance $p_{Cl_n}$ of the $\mathcal{W}_{Cl_n}$ is presented in Fig.\ \ref{fig:fig1}. It is observed that the $\mathcal{W}_{Cl_n}$ can outperform the best known class of EWs provided in the Ref.\ \cite{PhysRevLett.106.190502} for $n>5$. Meanwhile, the white noise tolerance $p_{Cl_n}$ exhibits a similar asymptotic behavior as in the first example, that is, tending to $1$ for large $n$. We remark that while the resulting EWs are quite robust, they are not optimal. In fact, the optimality of the bipartite EWs employed in the construction is not sufficient to guarantee the optimality of the resulting GME witness. For some explicit target states, one may either analytically or numerically optimize the result. While a systematic and operational improvement of this framework remains an open question. A brief discussion on this issue is provided at the end of Appendix.\ \ref{sec:appendix graph state}.

\subsection{Multipartite states admitting Schmidt decomposition} 
In the above examples, we benchmark our method with some well studied states. And now we turn to other less investigated states, where this method remains powerful. A typical class is the multipartite states admitting Schmidt decomposition. Without loss of generality, such state takes the form $|\phi_s\rangle=\sum_{i=0}^{d-1}\sqrt{\lambda_i}|i\rangle^{\otimes n}$, where the $\lambda_i$ are in decreasing order. This class of states includes high-dimensional GHZ states $ |GHZ_n^d\rangle = \sum_{i=0}^{d-1} |i\rangle^{\otimes n}/\sqrt{d}$ as a typical case when all the Schmidt coefficients are equal. For the multipartite states admitting Schmidt decomposition, our method leads to a class of optimal EWs (see Appendix.\ \ref{sec:appendix GHZ state} for a proof.)
\begin{equation}\label{eq:eg3}
	\mathcal{W}_{|\phi_s\rangle}=\sum_{\substack{i,j=0,\\i< j}}^{d-1} \sum_{r=1}^{n-1}\sum_m \sqrt{\lambda_i\lambda_j}\pi_m(|i\rangle^{\otimes r}|j\rangle^{\otimes n-r})\pi_m(\langle i^{\otimes r}|\langle j|^{\otimes n-r})+\sum_{i=0}^{d-1}\lambda_i|i\rangle\langle i|^{\otimes n}-|\phi_s\rangle\langle\phi_s|,
\end{equation}
where $\pi_m(|i\rangle^{\otimes r}|j\rangle^{\otimes n-r})$ is a permutation of $|i\rangle^{\otimes r}|j\rangle^{\otimes n-r}$
and the summation of $m$ is over all possible permutations. 
\begin{figure*}[h]
	\centering
	\includegraphics[width=15.83cm,height=12cm]{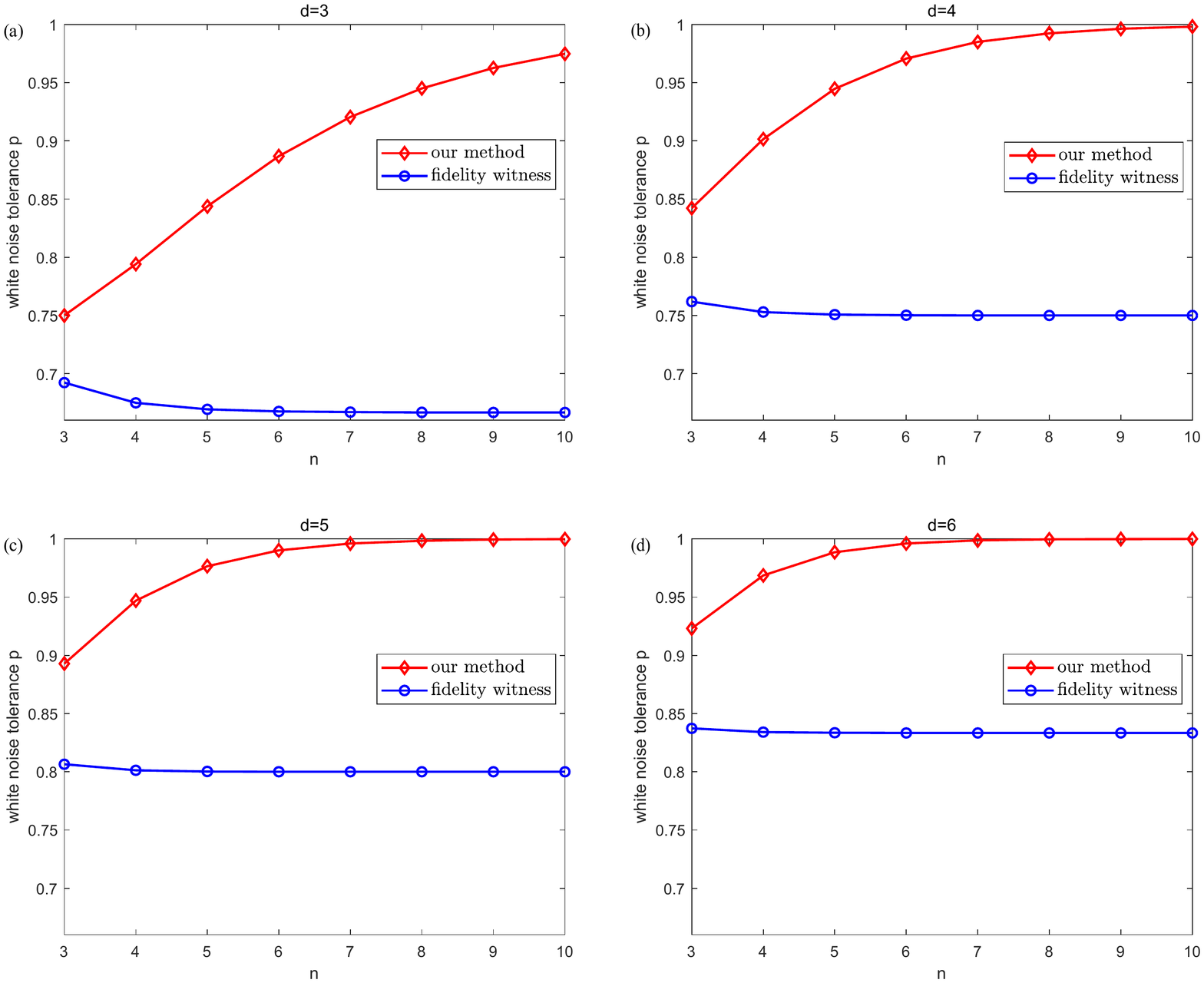}
	\caption{
	In sub-figure (a),(b),(c) and (d), we show how the white noise tolerance of different EWs varies with an increasing qudit number $n$, with $d=3,~4,~5,~6$ respectively. The target state is $d$-dimensional GHZ states, which belongs to the class of states in Example 1. In each sub-figure, the white noise tolerance of the GME witness in the Eq.\ (\ref{eq:eg3}) is plotted by red circled line, and in comparison the white noise tolerance of the normal fidelity witness $\mathcal{W}_F^{|\psi_s\rangle}=I/d -|GHZ_n^d\rangle\langle GHZ_n^d|$ is plotted by blue circled line. As observed in this figure, for all local dimension $d\ge 3$, our results increase and converge to $1$ with the increasing system size. While for the normal fidelity witness, its white noise tolerance decreases when the number of qudits grows and eventually tends to $1-1/d$.
	}
	\label{fig:ghz}
\end{figure*}
The white noise tolerance of $\mathcal{W}_{|\phi_s\rangle}$ is $p_{\mathcal{W}_{|\phi_s\rangle}}=(1-\sum_{i=0}^{d-1}\lambda_i^2)/(1-\sum_{i=0}^{d-1}\lambda_i^2+\frac{2^{n-1}-1}{d^n}((\sum_{i=0}^{d-1}\sqrt{\lambda_i})^2-1)))$. The $p_{\mathcal{W}_{|\phi_s\rangle}}$ tends to $1$ for large $n$ when $d > 2$. As a comparison, the best-known GME witness for this kind of states comes from the fidelity-based method, with $\mathcal{W}_F^{|\psi_s\rangle} = \lambda_0 I- |\phi_s\rangle\langle\phi_s|$. The white noise tolerance of $\mathcal{W}_F^{|\psi_s\rangle}$ is $(1-\lambda_0)/(1-1/d^n)$, which tends to $1-\lambda_0\le 1-1/d$ with an increasing system size. For the special case of $n$-qudit GHZ states $|GHZ_n^d\rangle$, the performance of our construction and the fidelity-based method is compared in Fig.\ \ref{fig:ghz}, where a significant improvement is demonstrated. Note that for $n$-qubit GHZ states $|GHZ_n\rangle$, the fidelity witness is already optimal, and hence we start from the local dimension $d=3$ in Fig.\ \ref{fig:ghz}.

\subsection{The four-qubit singlet state} 
Multi-qubit singlet states are another interesting family of multi-qubit states. They are invariant under a simultaneous unitary rotation on all qubits ($U^{\otimes n}|\psi\rangle\langle \psi| {U^{\dag}}^{\otimes n} =|\psi\rangle\langle \psi|$). In the four-qubit case, all four-qubit singlet states live in a two-dimensional subspace of the whole Hilbert space. Without loss of generality, it can be denoted as 
\begin{equation}
	|\varphi_4\rangle = a |\psi_{12}^-\rangle\otimes |\psi_{34}^-\rangle + e^{i\theta}b |\psi_{13}^-\rangle\otimes |\psi_{24}^-\rangle, 
\end{equation} 
with the constraint $a^2+b^2+cos(\theta)ab=1$ and $|\psi_{12}^-\rangle$ being the two-qubit singlet state $(|01\rangle-|10\rangle)/\sqrt{2}$ on the first two qubits. Specially, for the choice of $\theta = \pi/2$, one arrives at a class of four-qubit singlet states decided by a single parameter $|\varphi_4(a)\rangle = a |\psi_{12}^-\rangle\otimes |\psi_{34}^-\rangle + \sqrt{1-a^2} |\psi_{13}^-\rangle\otimes |\psi_{24}^-\rangle$ with $a \in [-1,1]$. For this class of state $|\varphi_4(a)\rangle$, our framework results in the following GME witness
\begin{equation}
	\mathcal{W}_4(a)= \alpha \mathcal{P}^4_2 +\frac{1}{2}(\mathcal{P}^4_1 + \mathcal{P}^4_3) +\frac{1}{4}(\mathcal{P}^4_0 + \mathcal{P}^4_4)- |\varphi_4\rangle\langle \varphi_4|,
\end{equation}
where $\alpha = \max\{(1+3(1-a^2))/4, (1+3a^2)/4\} \ge 5/8$. While the fidelity based witness for such state is $\mathcal{W}_4'(a) = \alpha I - |\varphi_4\rangle\langle \varphi_4|$. In Appendix.\ \ref{sec:appendix singlet}, a further discussion of the entanglement detection for multi-qubit singlet states is proposed based on our framework.

Consequently, we have provided a generic framework for detecting arbitrary target GME state in a noisy systems by constructing robust GME witnesses. Firstly, by benchmarking its performance on some well-studied states, it is observed that this framework results in robust GME witnesses that perform comparable to the current best witnesses for these states. For other less investigated states, the most widely used method to construct EW for them is the fidelity-based method. As shown in these examples, our framework can provide a significant improvement compared with the fidelity-based method. This also leads to the conjecture that a large amount of pure GME states become fairly robust to noise as the system size increases. Secondly, the advantage of our framework against the fidelity-based method comes with no experimental overheads. This benefits from the fact that the $\sum_k c_k \tilde{I}_k$ term in this construction is usually diagonal in some well-defined basis, such as the graph state basis and the computational basis. Finally, it should be stressed that Theorem 1 can be applied not only to the class of bipartite EWs shown in Eq.\ (\ref{eq:obEW}), but also to other classes of bipartite EWs. This potentially results in some different GME witnesses. Further example is provided in Appendix.\ \ref{sec:appendix generalization}.

\section{Applications of the resulting GME witnesses}

\subsection{Detection of unfaithfulness}
The unfaithful entangled states are a large class of states that cannot be recognized with any fidelity witness and have been attracted both theoretical and experimental interests \cite{PhysRevLett.126.140503,zhan2020detecting,PhysRevA.103.042417,PhysRevLett.127.220501}. Therefore, given that we have already gained access to construct finer GME witnesses than the fidelity-based method, it is natural to investigate their ability on the detection of unfaithful GME states. 

In general, deciding whether an entangled state is unfaithful is a nontrivial task, since one has to prove that the state is not detected by all fidelity witnesses, rather than a certain one. In bipartite case, a necessary and sufficient criterion for a state $\rho_{AB}$ to be unfaithful has been proposed \cite{PhysRevLett.126.140503}, while for multipartite case, it remains an open question on characterization of unfaithfulness. To avoid this difficulty and verify an EW indeed detects unfaithfulness, we limit our attention to a special class of states $\rho(p)=p I/d^n+(1-p) \rho_0$. that there is an upper bound on the white noise tolerance of any fidelity witness for arbitrary state. Denote $\mathcal{W}_F=\alpha I-\rho'$ as an arbitrary fidelity witness, then one can derive its white noise tolerance $p_F$ for arbitrary $\rho(p)$ by solving $Tr(\mathcal{W}_F\rho(p))=0$, which leads to
\begin{equation}
	p_F=\max\{\frac{Tr(\rho_0\rho')-\alpha}{Tr(\rho_0\rho')-1/d^n},0\}.
\end{equation}
Then it is straightforward to see that $p_F\le (1-1/d)/(1-1/d^n)$, due to the fact that $Tr(\rho_0\rho')\le 1$ and $\alpha \ge 1/d$. Hence it can be concluded that an EW can be employed to detect some unfaithful entangled states, as long as its white noise tolerance for some state is higher than $(1-1/d)/(1-1/d^n)$. This is precisely the case for many GME witnesses constructed with our framework. For example, in an $n$-qubit case, this upper bound is $1/(2(1-1/2^n))$ and decreases to $1/2$ as $n$ grows. While our framework provides large amount of EWs with white noise tolerance converging to $1$, allowing for the certification of unfaithfulness of many states in $n$-qubit case.

\subsection{Estimating entanglement measures}
Moreover, a witness operator is useful not only for entanglement certification, but also for entanglement quantification. To start with, we briefly review the method developed in Ref.\ \cite{PhysRevLett.98.110502} for optimally estimating some entanglement measure $E$ given the expectation value of some witness operator $\mathcal{W}$. The task can be described as finding the lower bound 
\begin{equation}
    \epsilon(w)=\inf_{\rho} \{E(\rho)|Tr(\rho \mathcal{W})=w\},
\end{equation}
where the infimum is taken over all states compatible with the data $w=Tr(\rho \mathcal{W})$. Note that $\epsilon(w)$ is a convex function, and thus there exist bounds of the type 
\begin{equation}
    \epsilon(w) \ge r\cdot w-c,
\end{equation}
for an arbitrary $r$. By inserting $w=Tr(\rho \mathcal{W})$ and $E(\rho) \ge \epsilon(w)$, it is observed that 
\begin{equation}
    c \ge r \cdot Tr(\rho \mathcal{W})-E(\rho),
\end{equation}
should be satisfied for any $\rho$. Hence given a "slope" $r$, the optimal constant $c$ is 
\begin{equation}
    c= \hat{E}(r\cdot \mathcal{W}):=\sup_{\rho} \{ r\cdot Tr(\rho \mathcal{W})-E(\rho)\}.
\end{equation}
Finally, an optimal lower bound is obtained after optimizing $r$:
\begin{equation}
    \epsilon(w)=\sup_r \{r\cdot w - \hat{E}(r\cdot \mathcal{W})\}.
\end{equation}
It should be remarked that we limit our discussions into the nontrivial case where a negative expectation $w$ of a witness operator is observed in the following. In this case, the optimal "slope" $r$ should always be negative. 

Now, suppose that the $\mathcal{W}_2$ is a finer EW than the $\mathcal{W}_1$, satisfying $\mathcal{W}_2 \preceq \mathcal{W}_1$. It is straightforward to see that 
\begin{equation}
    \hat{E}(r\cdot \mathcal{W}_1) \ge \hat{E}(r\cdot \mathcal{W}_2).
\end{equation}
Therefore, when these two operators $\mathcal{W}_1$ and $\mathcal{W}_2$ have the same expectation value $w_0$, 
\begin{equation}
	\begin{aligned}
    	\epsilon_2(w_0)= &\sup_r \{r\cdot w_0 - \hat{E}(r\cdot \mathcal{W}_2)\}  \\
    	\ge &\sup_r \{r\cdot w_0 - \hat{E}(r\cdot \mathcal{W}_1)\} \\
    	=& \epsilon_1(w_0).
	\end{aligned}
\end{equation}
For the same target state $\rho$, the expectations $w_1$ and $w_2$ of these two witness operators always satisfy $w_1 \ge w_2$, which leads to $\epsilon_2(w_2) \ge \epsilon_2(w_1) \ge \epsilon_1(w_1)$. That is, a finer EW provides a tighter lower bound on the entanglement measure for the same state. Hence, our framework enables a better estimation of the entanglement measures of genuine multipartite entanglement than the fidelity-based method. 

To quantitatively investigate the improvement from these new GME witness, we discuss the estimation on the geometry measure of genuine multipartite entanglement for noisy $n$-partite $d$-dimensional GHZ states $\rho_{n,d}(p)=pI/d^n+(1-p) |GHZ_n^d\rangle\langle GHZ_n^d|$, with $ |GHZ_n^d\rangle = \sum_{i=0}^{d-1} |i\rangle^{\otimes n}/\sqrt{d}$. For arbitrary multipartite pure state $|\psi\rangle$, its geometric measurement of GME is defined by $E_G(|\psi\rangle) = 1-\max_{|\phi_{bs}\rangle}|\langle\phi_{bs}|\psi\rangle|^2$, with $|\phi_{bs}\rangle$ being arbitrary pure biseparable state. The geometric measure of GME is extended to mixed states by the convex roof construction
\begin{equation}
	E_G(\rho)=\inf \limits_{p_i,|\psi_i\rangle}\sum_i p_i E_G(|\psi_i\rangle),
\end{equation}
where the minimization runs over all possible decompositions $\rho=\sum_i p_i|\psi_i\rangle\langle\psi_i|$. 

Based on the result in Ref.\ \cite{PhysRevApplied.13.054022}, one can derive a lower bound $\epsilon_f^{n,d}(p)$ for $E_G(\rho_{n,d}(p))$
\begin{equation}
	E_G(\rho_{n,d}(p)) \ge \epsilon_f^{n,d}(p) :=1-\gamma(S),
\end{equation}
where $\gamma(S)=[\sqrt{S}+\sqrt{(d-1)(d-S)}]^2/d$ with $S=\max \{1, d(1-p)+p/d^{n-1}\}$. This is just the lower bound related to the fidelity witness $\mathcal{W}_F=I/d-|GHZ_n^d\rangle\langle GHZ_n^d|$. Whereas it has been proved in the previous section that finer EW is accessible with our method, that is,
\begin{equation}
	\begin{aligned}
		\mathcal{W}_{o,|GHZ_n^d\rangle}=&\sum_{\substack{i,j=0,\\i< j}}^{d-1} \sum_{r=1}^{n-1} \sum_m \frac{1}{d}\pi_m(|i\rangle^{\otimes r}|j\rangle^{\otimes n-r})\pi_m(\langle i^{\otimes r}|\langle j|^{\otimes n-r}) \\
		&+\sum_{i=0}^{d-1}\frac{1}{d}|i\rangle\langle i|^{\otimes n}-|GHZ_n^d\rangle\langle GHZ_n^d|.
	\end{aligned}
\end{equation}
With the expectation value $w_{n,d}(p)=Tr(\rho_{n,d}(p)\mathcal{W}_{o,|GHZ_n^d\rangle})$ from this finer EW, a lower bound $\epsilon_o^{n,d}(p)$ can be derived by employing the technique developed in Ref.\ \cite{PhysRevLett.98.110502}:
\begin{equation} \label{lowerbd}
	E_G(\rho_{n,d}(p)) \ge \epsilon_o^{n,d}(p) :=\sup \limits_{r}\left\{r\cdot w_{n,d}(p)-\hat{E}_G(r\mathcal{W}_{o,|GHZ_n^d\rangle})\right\},
\end{equation}
with $r$ being a real number, and
\begin{equation}\label{hatE}
	\hat{E}_G(r\mathcal{W}_{o,|GHZ_n^d\rangle})=\sup\limits_{|\psi\rangle}\sup\limits_{|\phi_{bs}\rangle} \left\{\langle\psi|(r\mathcal{W}_{o,|GHZ_n^d\rangle}+|\phi_{bs}\rangle\langle\phi_{bs}|)|\psi\rangle-1\right\},
\end{equation}
where the maximization runs over all pure state $|\psi\rangle$ and biseparable state $|\phi_{bs}\rangle$. Furthermore, in this special case, it can be verified that one has to choose $|\phi_{bs}\rangle$ as a state having the largest overlap with $|GHZ_n^d\rangle$, which results in
\begin{equation}
	\hat{E}_G(r\mathcal{W}_{o,|GHZ_n^d\rangle})=\frac{1-r}{2}+\frac{1}{2}\sqrt{(1-r)^2+4r\frac{d-1}{d}}+\frac{r}{d}-1.
\end{equation}
By inserting this equation into Eq.\ (\ref{lowerbd}), the lower bound $\epsilon_o^{n,d}(p)$ can be solved directly.

\begin{figure*}[h]
	\centering
	\includegraphics[width=14.85cm,height=11.15cm]{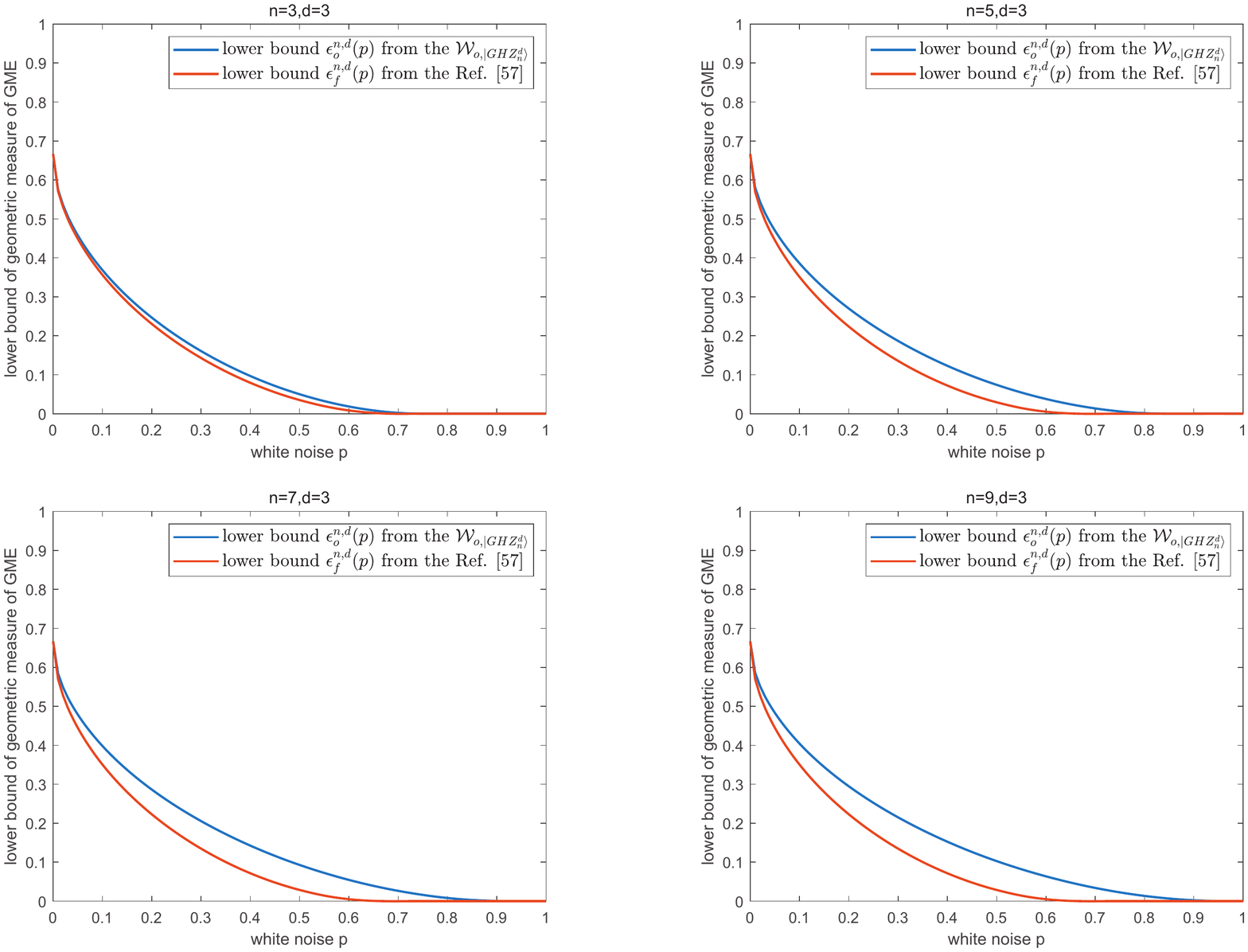}
	\caption{
	In this figure, we choose $d=3$ and $n=3,~5,~7,~9$ as examples to compare the lower bound $\epsilon_o^{n,d}(p)$ and $\epsilon_f^{n,d}(p)$ from different EWs. It is observed that the advantage of EWs in this work become more apparent with increasing system size. Note that for higher local dimensions where $d\ge 4$, the lower bound $\epsilon_o^{n,d}(p)$ and $\epsilon_f^{n,d}(p)$ have similar behavior. While in the qubit case, the EW $\mathcal{W}_{o,|GHZ_n^d\rangle}$ degenerates to the fidelity witness $I/2-|GHZ_n\rangle\langle GHZ_n|$ for $n$-qubit GHZ state $|GHZ_n\rangle$, and $\epsilon_o^{n,2}(p)$ is the same as $\epsilon_f^{n,2}(p)$.
	}
	\label{fig:measure}
\end{figure*}

In Fig.\ \ref{fig:measure}, we have shown the results for $d=3$ and $n=3,~5,~7,~9$ as examples, to illustrate the performance of our method with an increasing system size. As the number of subsystems grows, the critical value of $p$ tends to $1$, when the lower bound $\epsilon_o^{n,d}(p)$ vanishes. Meanwhile, the $\epsilon_o^{n,d}(p)$ is always larger than the $\epsilon_f^{n,d}(p)$ above, which vanishes at $p=1-1/d$ for large $n$. That is, the new EWs $\mathcal{W}_{o,|GHZ_n^d\rangle}$ are able to provide a better estimation on the geometric measure of GME for $\rho_{n,d}(p)$. It remains open whether $\epsilon_o^{n,d}(p)$ equals $E_G(\rho_{n,d}(p))$. However, it is still reasonable to expect that such new GME witnesses can provide faithful estimations on entanglement measures without the need for quantum tomography, as they are already robust enough.

\section{Conclusion and outlook}

In summary, we have developed a exquisite framework and scheme for genuine multipartite entanglement detection, and demonstrated its operability and universality by applying it on typical GME states that arise in practice. In particular, this is achieved using a novel method to bring any complete set of bipartite EWs to a single GME witness. This method allows to make full use of some prior information about the target state to improve the noise resistance. In fact, the resulting GME witnesses turn out to be quite robust, whose white noise tolerance converge to $1$ in many cases. As a consequence, this framework holds great practical potential in real-life situations, especially for detecting entanglement in noisy multipartite or high-dimensional systems. This will play a certain role in facilitating the solution of the very challenging problem of genuine multipartite entanglement detection.

In addition to genuine multipartite entanglement, we remark that our method is highly flexible and admits natural generalizations for detecting other types of entanglement. A relevant case is entanglement detection in quantum networks, which is currently under active investigations. In quantum networks, multipartite entanglement exhibits novel features due to the complex network topology \cite{PhysRevLett.125.240505,PhysRevA.103.L060401,PhysRevLett.128.220501}, and better techniques are urgently needed for the characterization of genuine network multipartite entanglement. Finally, it will also be interesting to seek for further extension of the framework in high-order entanglement detection \cite{PhysRevLett.105.210504} as well as bound entanglement detection.

\section*{Acknowledgments}
We thank Yi-Zheng Zhen for very valuable discussion. This work has been supported by the National Natural Science Foundation of China (Grants No. 62031024, 11874346, 12174375), the National Key R$\&$D Program of China (2019YFA0308700), the Anhui Initiative in Quantum Information Technologies (AHY060200), and the Innovation Program for Quantum Science and Technology (No. 2021ZD0301100).

%%%%%%%%%%%%%%%%%%%%%%%%%%%%%%%%%%%%%%%%
% choose a .bib file
\bibliographystyle{quantum}
\bibliography{EWbeyondFidelity.bib}

\begin{thebibliography}{10}

\bibitem{RevModPhys.81.865}
Ryszard Horodecki, Pawe\l{} Horodecki, Micha\l{} Horodecki, and Karol
  Horodecki.
\newblock ``Quantum entanglement''.
\newblock \href{https://dx.doi.org/10.1103/RevModPhys.81.865}{Rev. Mod. Phys.
  {\bf 81}, 865--942}~(2009).

\bibitem{RevModPhys.91.025001}
Eric Chitambar and Gilad Gour.
\newblock ``Quantum resource theories''.
\newblock \href{https://dx.doi.org/10.1103/RevModPhys.91.025001}{Rev. Mod.
  Phys. {\bf 91}, 025001}~(2019).

\bibitem{bennett2000quantum}
Charles~H Bennett and David~P DiVincenzo.
\newblock ``Quantum information and computation''.
\newblock \href{https://dx.doi.org/10.1038/35005001}{Nature {\bf 404},
  247--255}~(2000).

\bibitem{PhysRevLett.70.1895}
Charles~H. Bennett, Gilles Brassard, Claude Cr\'epeau, Richard Jozsa, Asher
  Peres, and William~K. Wootters.
\newblock ``Teleporting an unknown quantum state via dual classical and
  einstein-podolsky-rosen channels''.
\newblock \href{https://dx.doi.org/10.1103/PhysRevLett.70.1895}{Phys. Rev.
  Lett. {\bf 70}, 1895--1899}~(1993).

\bibitem{RevModPhys.81.1301}
Valerio Scarani, Helle Bechmann-Pasquinucci, Nicolas~J. Cerf, Miloslav
  Du\ifmmode~\check{s}\else \v{s}\fi{}ek, Norbert L\"utkenhaus, and Momtchil
  Peev.
\newblock ``The security of practical quantum key distribution''.
\newblock \href{https://dx.doi.org/10.1103/RevModPhys.81.1301}{Rev. Mod. Phys.
  {\bf 81}, 1301--1350}~(2009).

\bibitem{feynman1982simulating}
Richard~P Feynman.
\newblock ``Simulating physics with computers''.
\newblock \href{https://dx.doi.org/10.1007/BF02650179}{Int. J. Theor. Phys {\bf
  21}, 467--488}~(1982).

\bibitem{deutsch1985quantum}
David Deutsch.
\newblock ``Quantum theory, the {Church--Turing} principle and the universal
  quantum computer''.
\newblock \href{https://dx.doi.org/10.1098/rspa.1985.0070}{Proc. R. Soc. London
  A {\bf 400}, 97--117}~(1985).

\bibitem{Luo620}
Xin-Yu Luo, Yi-Quan Zou, Ling-Na Wu, Qi~Liu, Ming-Fei Han, Meng~Khoon Tey, and
  Li~You.
\newblock ``Deterministic entanglement generation from driving through quantum
  phase transitions''.
\newblock \href{https://dx.doi.org/10.1126/science.aag1106}{Science {\bf 355},
  620--623}~(2017).

\bibitem{arute2019quantum}
Frank Arute, Kunal Arya, Ryan Babbush, Dave Bacon, Joseph~C Bardin, Rami
  Barends, Rupak Biswas, Sergio Boixo, Fernando~GSL Brandao, David~A Buell,
  et~al.
\newblock ``Quantum supremacy using a programmable superconducting processor''.
\newblock \href{https://dx.doi.org/10.5061/dryad.k6t1rj8}{Nature {\bf 574},
  505--510}~(2019).

\bibitem{PhysRevLett.105.210504}
Che-Ming Li, Kai Chen, Andreas Reingruber, Yueh-Nan Chen, and Jian-Wei Pan.
\newblock ``Verifying genuine high-order entanglement''.
\newblock \href{https://dx.doi.org/10.1103/PhysRevLett.105.210504}{Phys. Rev.
  Lett. {\bf 105}, 210504}~(2010).

\bibitem{yao2012}
Xing-Can Yao, Tian-Xiong Wang, Ping Xu, He~Lu, Ge-Sheng Pan, Xiao-Hui Bao,
  Cheng-Zhi Peng, Chao-Yang Lu, Yu-Ao Chen, and Jian-Wei Pan.
\newblock ``Observation of eight-photon entanglement''.
\newblock \href{https://dx.doi.org/10.1038/nphoton.2011.354}{Nat. Photonics
  {\bf 6}, 225--228}~(2012).

\bibitem{doi:10.1126/science.abg7812}
Ming Gong, Shiyu Wang, Chen Zha, Ming-Cheng Chen, He-Liang Huang, Yulin Wu,
  Qingling Zhu, Youwei Zhao, Shaowei Li, Shaojun Guo, Haoran Qian, Yangsen Ye,
  Fusheng Chen, Chong Ying, Jiale Yu, Daojin Fan, Dachao Wu, Hong Su, Hui Deng,
  Hao Rong, Kaili Zhang, Sirui Cao, Jin Lin, Yu~Xu, Lihua Sun, Cheng Guo,
  Na~Li, Futian Liang, V.~M. Bastidas, Kae Nemoto, W.~J. Munro, Yong-Heng Huo,
  Chao-Yang Lu, Cheng-Zhi Peng, Xiaobo Zhu, and Jian-Wei Pan.
\newblock ``Quantum walks on a programmable two-dimensional 62-qubit
  superconducting processor''.
\newblock \href{https://dx.doi.org/10.1126/science.abg7812}{Science {\bf 372},
  948--952}~(2021).

\bibitem{yokoyama2013ultra}
Shota Yokoyama, Ryuji Ukai, Seiji~C Armstrong, Chanond Sornphiphatphong,
  Toshiyuki Kaji, Shigenari Suzuki, Jun-ichi Yoshikawa, Hidehiro Yonezawa,
  Nicolas~C Menicucci, and Akira Furusawa.
\newblock ``Ultra-large-scale continuous-variable cluster states multiplexed in
  the time domain''.
\newblock \href{https://dx.doi.org/10.1038/nphoton.2013.287}{Nature Photonics
  {\bf 7}, 982--986}~(2013).

\bibitem{PhysRevLett.112.155304}
Bernd L\"ucke, Jan Peise, Giuseppe Vitagliano, Jan Arlt, Luis Santos, G\'eza
  T\'oth, and Carsten Klempt.
\newblock ``Detecting multiparticle entanglement of dicke states''.
\newblock \href{https://dx.doi.org/10.1103/PhysRevLett.112.155304}{Phys. Rev.
  Lett. {\bf 112}, 155304}~(2014).

\bibitem{Preskill2018quantumcomputingin}
John Preskill.
\newblock ``Quantum {C}omputing in the {NISQ} era and beyond''.
\newblock \href{https://dx.doi.org/10.22331/q-2018-08-06-79}{{Quantum} {\bf 2},
  79}~(2018).

\bibitem{PhysRevLett.77.1413}
Asher Peres.
\newblock ``Separability criterion for density matrices''.
\newblock \href{https://dx.doi.org/10.1103/PhysRevLett.77.1413}{Phys. Rev.
  Lett. {\bf 77}, 1413--1415}~(1996).

\bibitem{HORODECKI1997333}
Pawel Horodecki.
\newblock ``Separability criterion and inseparable mixed states with positive
  partial transposition''.
\newblock
  \href{https://dx.doi.org/https://doi.org/10.1016/S0375-9601(97)00416-7}{Phys.
  Lett. A {\bf 232}, 333--339}~(1997).

\bibitem{chen205017quantum}
Kai Chen and Ling-An Wu.
\newblock ``A matrix realignment method for recognizing entanglement''.
\newblock \href{https://dx.doi.org/10.26421/QIC3.3-1}{Quantum Inform. Comput.
  {\bf 3}, 193--202}~(2003).

\bibitem{rudolph2005further}
Oliver Rudolph.
\newblock ``Further results on the cross norm criterion for separability''.
\newblock \href{https://dx.doi.org/10.1007/s11128-005-5664-1}{Quantum Inf.
  Process. {\bf 4}, 219--239}~(2005).

\bibitem{PhysRevLett.95.040504}
Kai Chen, Sergio Albeverio, and Shao-Ming Fei.
\newblock ``Concurrence of arbitrary dimensional bipartite quantum states''.
\newblock \href{https://dx.doi.org/10.1103/PhysRevLett.95.040504}{Phys. Rev.
  Lett. {\bf 95}, 040504}~(2005).

\bibitem{PhysRevLett.95.210501}
Kai Chen, Sergio Albeverio, and Shao-Ming Fei.
\newblock ``Entanglement of formation of bipartite quantum states''.
\newblock \href{https://dx.doi.org/10.1103/PhysRevLett.95.210501}{Phys. Rev.
  Lett. {\bf 95}, 210501}~(2005).

\bibitem{plenio2014introduction}
Martin~B Plenio and Shashank~S Virmani.
\newblock ``An introduction to entanglement theory''.
\newblock \href{https://dx.doi.org/10.48550/arXiv.quant-ph/0504163}{Quantum
  Inform. Comput. {\bf 7}, 1--51}~(2007).

\bibitem{GUHNE20091}
Otfried Gühne and Géza Tóth.
\newblock ``Entanglement detection''.
\newblock
  \href{https://dx.doi.org/https://doi.org/10.1016/j.physrep.2009.02.004}{Phys.
  Rep. {\bf 474}, 1--75}~(2009).

\bibitem{friis2019entanglement}
Nicolai Friis, Giuseppe Vitagliano, Mehul Malik, and Marcus Huber.
\newblock ``Entanglement certification from theory to experiment''.
\newblock \href{https://dx.doi.org/10.1038/s42254-018-0003-5}{Nat. Rev. Phys.
  {\bf 1}, 72--87}~(2019).

\bibitem{PhysRevLett.108.110501}
Nicolas Brunner, James Sharam, and Tam\'as V\'ertesi.
\newblock ``Testing the structure of multipartite entanglement with bell
  inequalities''.
\newblock \href{https://dx.doi.org/10.1103/PhysRevLett.108.110501}{Phys. Rev.
  Lett. {\bf 108}, 110501}~(2012).

\bibitem{zhou2019detecting}
You Zhou, Qi~Zhao, Xiao Yuan, and Xiongfeng Ma.
\newblock ``Detecting multipartite entanglement structure with minimal
  resources''.
\newblock \href{https://dx.doi.org/10.1038/s41534-019-0200-9}{npj Quantum
  Information {\bf 5}, 1--8}~(2019).

\bibitem{HORODECKI19961}
Michał Horodecki, Paweł Horodecki, and Ryszard Horodecki.
\newblock ``Separability of mixed states: necessary and sufficient
  conditions''.
\newblock
  \href{https://dx.doi.org/https://doi.org/10.1016/S0375-9601(96)00706-2}{Phys.
  Lett. A {\bf 223}, 1--8}~(1996).

\bibitem{terhal2000bell}
Barbara~M Terhal.
\newblock ``Bell inequalities and the separability criterion''.
\newblock \href{https://dx.doi.org/10.1016/S0375-9601(00)00401-1}{Phys. Lett. A
  {\bf 271}, 319--326}~(2000).

\bibitem{lewenstein2000optimization}
Maciej Lewenstein, Barabara Kraus, J~Ignacio Cirac, and P~Horodecki.
\newblock ``Optimization of entanglement witnesses''.
\newblock \href{https://dx.doi.org/10.1103/PhysRevA.62.052310}{Phys. Rev. A
  {\bf 62}, 052310}~(2000).

\bibitem{hyllus2005relations}
Philipp Hyllus, Otfried G{\"u}hne, Dagmar Bru{\ss}, and Maciej Lewenstein.
\newblock ``Relations between entanglement witnesses and bell inequalities''.
\newblock \href{https://dx.doi.org/10.1103/PhysRevA.72.012321}{Phys. Rev. A
  {\bf 72}, 012321}~(2005).

\bibitem{lewenstein2001characterization}
Maciej Lewenstein, B~Kraus, P~Horodecki, and JI~Cirac.
\newblock ``Characterization of separable states and entanglement witnesses''.
\newblock \href{https://dx.doi.org/10.1103/PhysRevA.63.044304}{Phys. Rev. A
  {\bf 63}, 044304}~(2001).

\bibitem{PhysRevLett.98.110502}
O.~G\"uhne, M.~Reimpell, and R.~F. Werner.
\newblock ``Estimating entanglement measures in experiments''.
\newblock \href{https://dx.doi.org/10.1103/PhysRevLett.98.110502}{Phys. Rev.
  Lett. {\bf 98}, 110502}~(2007).

\bibitem{PhysRevX.8.021012}
Nicolai Friis, Oliver Marty, Christine Maier, Cornelius Hempel, Milan
  Holz\"apfel, Petar Jurcevic, Martin~B. Plenio, Marcus Huber, Christian Roos,
  Rainer Blatt, and Ben Lanyon.
\newblock ``Observation of entangled states of a fully controlled 20-qubit
  system''.
\newblock \href{https://dx.doi.org/10.1103/PhysRevX.8.021012}{Phys. Rev. X {\bf
  8}, 021012}~(2018).

\bibitem{PhysRevLett.106.130506}
Thomas Monz, Philipp Schindler, Julio~T. Barreiro, Michael Chwalla, Daniel
  Nigg, William~A. Coish, Maximilian Harlander, Wolfgang H\"ansel, Markus
  Hennrich, and Rainer Blatt.
\newblock ``14-qubit entanglement: Creation and coherence''.
\newblock \href{https://dx.doi.org/10.1103/PhysRevLett.106.130506}{Phys. Rev.
  Lett. {\bf 106}, 130506}~(2011).

\bibitem{PhysRevLett.95.210502}
Nikolai Kiesel, Christian Schmid, Ulrich Weber, G\'eza T\'oth, Otfried G\"uhne,
  Rupert Ursin, and Harald Weinfurter.
\newblock ``Experimental analysis of a four-qubit photon cluster state''.
\newblock \href{https://dx.doi.org/10.1103/PhysRevLett.95.210502}{Phys. Rev.
  Lett. {\bf 95}, 210502}~(2005).

\bibitem{gao2010experimental}
Wei-Bo Gao, Chao-Yang Lu, Xing-Can Yao, Ping Xu, Otfried G{\"u}hne, Alexander
  Goebel, Yu-Ao Chen, Cheng-Zhi Peng, Zeng-Bing Chen, and Jian-Wei Pan.
\newblock ``Experimental demonstration of a hyper-entangled ten-qubit
  schr{\"o}dinger cat state''.
\newblock \href{https://dx.doi.org/10.1038/nphys1603}{Nat. Phys. {\bf 6},
  331--335}~(2010).

\bibitem{PhysRevX.8.021072}
He~Lu, Qi~Zhao, Zheng-Da Li, Xu-Fei Yin, Xiao Yuan, Jui-Chen Hung, Luo-Kan
  Chen, Li~Li, Nai-Le Liu, Cheng-Zhi Peng, Yeong-Cherng Liang, Xiongfeng Ma,
  Yu-Ao Chen, and Jian-Wei Pan.
\newblock ``Entanglement structure: Entanglement partitioning in multipartite
  systems and its experimental detection using optimizable witnesses''.
\newblock \href{https://dx.doi.org/10.1103/PhysRevX.8.021072}{Phys. Rev. X {\bf
  8}, 021072}~(2018).

\bibitem{PhysRevLett.124.160503}
Zheng-Da Li, Qi~Zhao, Rui Zhang, Li-Zheng Liu, Xu-Fei Yin, Xingjian Zhang,
  Yue-Yang Fei, Kai Chen, Nai-Le Liu, Feihu Xu, Yu-Ao Chen, Li~Li, and Jian-Wei
  Pan.
\newblock ``Measurement-device-independent entanglement witness of tripartite
  entangled states and its applications''.
\newblock \href{https://dx.doi.org/10.1103/PhysRevLett.124.160503}{Phys. Rev.
  Lett. {\bf 124}, 160503}~(2020).

\bibitem{PhysRevLett.122.110501}
Ming Gong, Ming-Cheng Chen, Yarui Zheng, Shiyu Wang, Chen Zha, Hui Deng,
  Zhiguang Yan, Hao Rong, Yulin Wu, Shaowei Li, Fusheng Chen, Youwei Zhao,
  Futian Liang, Jin Lin, Yu~Xu, Cheng Guo, Lihua Sun, Anthony~D. Castellano,
  Haohua Wang, Chengzhi Peng, Chao-Yang Lu, Xiaobo Zhu, and Jian-Wei Pan.
\newblock ``Genuine 12-qubit entanglement on a superconducting quantum
  processor''.
\newblock \href{https://dx.doi.org/10.1103/PhysRevLett.122.110501}{Phys. Rev.
  Lett. {\bf 122}, 110501}~(2019).

\bibitem{Greenberger1989}
Daniel~M. Greenberger, Michael~A. Horne, and Anton Zeilinger.
\newblock ``Going beyond bell's theorem''.
\newblock \href{https://dx.doi.org/10.1007/978-94-017-0849-4_10}{Pages 69--72}.
\newblock Springer Netherlands. Dordrecht~(1989).

\bibitem{PhysRevA.62.062314}
W.~D\"ur, G.~Vidal, and J.~I. Cirac.
\newblock ``Three qubits can be entangled in two inequivalent ways''.
\newblock \href{https://dx.doi.org/10.1103/PhysRevA.62.062314}{Phys. Rev. A
  {\bf 62}, 062314}~(2000).

\bibitem{PhysRevA.69.062311}
M.~Hein, J.~Eisert, and H.~J. Briegel.
\newblock ``Multiparty entanglement in graph states''.
\newblock \href{https://dx.doi.org/10.1103/PhysRevA.69.062311}{Phys. Rev. A
  {\bf 69}, 062311}~(2004).

\bibitem{hein2006entanglement}
M.~Hein, W.~Dür, J.~Eisert, R.~Raussendorf, M.~Van den Nest, and H.~J.
  Briegel.
\newblock ``Entanglement in graph states and its applications''~(2006).
\newblock
  \href{http://arxiv.org/abs/quant-ph/0602096}{arXiv:quant-ph/0602096}.

\bibitem{PhysRevLett.92.087902}
Mohamed Bourennane, Manfred Eibl, Christian Kurtsiefer, Sascha Gaertner, Harald
  Weinfurter, Otfried G\"uhne, Philipp Hyllus, Dagmar Bru\ss{}, Maciej
  Lewenstein, and Anna Sanpera.
\newblock ``Experimental detection of multipartite entanglement using witness
  operators''.
\newblock \href{https://dx.doi.org/10.1103/PhysRevLett.92.087902}{Phys. Rev.
  Lett. {\bf 92}, 087902}~(2004).

\bibitem{PhysRevLett.106.190502}
Bastian Jungnitsch, Tobias Moroder, and Otfried G\"uhne.
\newblock ``Taming multiparticle entanglement''.
\newblock \href{https://dx.doi.org/10.1103/PhysRevLett.106.190502}{Phys. Rev.
  Lett. {\bf 106}, 190502}~(2011).

\bibitem{PhysRevLett.113.100501}
Marcus Huber and Ritabrata Sengupta.
\newblock ``Witnessing genuine multipartite entanglement with positive maps''.
\newblock \href{https://dx.doi.org/10.1103/PhysRevLett.113.100501}{Phys. Rev.
  Lett. {\bf 113}, 100501}~(2014).

\bibitem{PhysRevLett.111.110503}
J.~Sperling and W.~Vogel.
\newblock ``Multipartite entanglement witnesses''.
\newblock \href{https://dx.doi.org/10.1103/PhysRevLett.111.110503}{Phys. Rev.
  Lett. {\bf 111}, 110503}~(2013).

\bibitem{PhysRevLett.124.200502}
M.~Weilenmann, B.~Dive, D.~Trillo, E.~A. Aguilar, and M.~Navascu\'es.
\newblock ``Entanglement detection beyond measuring fidelities''.
\newblock \href{https://dx.doi.org/10.1103/PhysRevLett.124.200502}{Phys. Rev.
  Lett. {\bf 124}, 200502}~(2020).

\bibitem{Bergmann_2013}
Marcel Bergmann and Otfried Gühne.
\newblock ``Entanglement criteria for dicke states''.
\newblock \href{https://dx.doi.org/10.1088/1751-8113/46/38/385304}{J. Phys. A:
  Math. Theor {\bf 46}, 385304}~(2013).

\bibitem{briegel2009measurement}
Hans~J Briegel, David~E Browne, Wolfgang D{\"u}r, Robert Raussendorf, and
  Maarten Van~den Nest.
\newblock ``Measurement-based quantum computation''.
\newblock \href{https://dx.doi.org/10.1038/nphys1157}{Nat. Phys. {\bf 5},
  19--26}~(2009).

\bibitem{PhysRevA.65.012308}
D.~Schlingemann and R.~F. Werner.
\newblock ``Quantum error-correcting codes associated with graphs''.
\newblock \href{https://dx.doi.org/10.1103/PhysRevA.65.012308}{Phys. Rev. A
  {\bf 65}, 012308}~(2001).

\bibitem{PhysRevA.84.032310}
Bastian Jungnitsch, Tobias Moroder, and Otfried G\"uhne.
\newblock ``Entanglement witnesses for graph states: General theory and
  examples''.
\newblock \href{https://dx.doi.org/10.1103/PhysRevA.84.032310}{Phys. Rev. A
  {\bf 84}, 032310}~(2011).

\bibitem{PhysRevLett.126.140503}
Otfried G\"uhne, Yuanyuan Mao, and Xiao-Dong Yu.
\newblock ``Geometry of faithful entanglement''.
\newblock \href{https://dx.doi.org/10.1103/PhysRevLett.126.140503}{Phys. Rev.
  Lett. {\bf 126}, 140503}~(2021).

\bibitem{zhan2020detecting}
Yongtao Zhan and Hoi-Kwong Lo.
\newblock ``Detecting entanglement in unfaithful states''~(2020).
\newblock  \href{http://arxiv.org/abs/2010.06054}{arXiv:2010.06054}.

\bibitem{PhysRevA.103.042417}
Gabriele Riccardi, Daniel~E. Jones, Xiao-Dong Yu, Otfried G\"uhne, and Brian~T.
  Kirby.
\newblock ``Exploring the relationship between the faithfulness and
  entanglement of two qubits''.
\newblock \href{https://dx.doi.org/10.1103/PhysRevA.103.042417}{Phys. Rev. A
  {\bf 103}, 042417}~(2021).

\bibitem{PhysRevLett.127.220501}
Xiao-Min Hu, Wen-Bo Xing, Yu~Guo, Mirjam Weilenmann, Edgar~A. Aguilar, Xiaoqin
  Gao, Bi-Heng Liu, Yun-Feng Huang, Chuan-Feng Li, Guang-Can Guo, Zizhu Wang,
  and Miguel Navascu\'es.
\newblock ``Optimized detection of high-dimensional entanglement''.
\newblock \href{https://dx.doi.org/10.1103/PhysRevLett.127.220501}{Phys. Rev.
  Lett. {\bf 127}, 220501}~(2021).

\bibitem{PhysRevApplied.13.054022}
Yue Dai, Yuli Dong, Zhenyu Xu, Wenlong You, Chengjie Zhang, and Otfried
  G\"uhne.
\newblock ``Experimentally accessible lower bounds for genuine multipartite
  entanglement and coherence measures''.
\newblock \href{https://dx.doi.org/10.1103/PhysRevApplied.13.054022}{Phys. Rev.
  Applied {\bf 13}, 054022}~(2020).

\bibitem{PhysRevLett.125.240505}
Miguel Navascu\'es, Elie Wolfe, Denis Rosset, and Alejandro Pozas-Kerstjens.
\newblock ``Genuine network multipartite entanglement''.
\newblock \href{https://dx.doi.org/10.1103/PhysRevLett.125.240505}{Phys. Rev.
  Lett. {\bf 125}, 240505}~(2020).

\bibitem{PhysRevA.103.L060401}
Tristan Kraft, S\'ebastien Designolle, Christina Ritz, Nicolas Brunner, Otfried
  G\"uhne, and Marcus Huber.
\newblock ``Quantum entanglement in the triangle network''.
\newblock \href{https://dx.doi.org/10.1103/PhysRevA.103.L060401}{Phys. Rev. A
  {\bf 103}, L060401}~(2021).

\bibitem{PhysRevLett.128.220501}
Patricia Contreras-Tejada, Carlos Palazuelos, and Julio~I. de~Vicente.
\newblock ``Asymptotic survival of genuine multipartite entanglement in noisy
  quantum networks depends on the topology''.
\newblock \href{https://dx.doi.org/10.1103/PhysRevLett.128.220501}{Phys. Rev.
  Lett. {\bf 128}, 220501}~(2022).

\bibitem{PhysRevA.62.052310}
M.~Lewenstein, B.~Kraus, J.~I. Cirac, and P.~Horodecki.
\newblock ``Optimization of entanglement witnesses''.
\newblock \href{https://dx.doi.org/10.1103/PhysRevA.62.052310}{Phys. Rev. A
  {\bf 62}, 052310}~(2000).

\bibitem{Augusiak_2011}
R~Augusiak, J~Tura, and M~Lewenstein.
\newblock ``A note on the optimality of decomposable entanglement witnesses and
  completely entangled subspaces''.
\newblock \href{https://dx.doi.org/10.1088/1751-8113/44/21/212001}{J. Phys. A:
  Math. Theor. {\bf 44}, 212001}~(2011).

\bibitem{PhysRevA.59.141}
Guifr\'e Vidal and Rolf Tarrach.
\newblock ``Robustness of entanglement''.
\newblock \href{https://dx.doi.org/10.1103/PhysRevA.59.141}{Phys. Rev. A {\bf
  59}, 141--155}~(1999).

\end{thebibliography}
%%%%%%%%%%%%%%%%%%%%%%%%%%%%%%%%%%%%%%%%

\appendix
\section*{Appendix}

\section{Proof and discussions of the bipartite EW in Eq. (5)}\label{sec:appendix bipartite EW} 

\subsection {A class of bipartite entanglement witness}

Let $|\phi\rangle$ be an arbitrary pure entangled state in the $d\times d$ dimensional Hilbert space $\mathcal{H}_d\otimes\mathcal{H}_d$. Without loss of generality, one can assume $|\phi\rangle=\sum_{i=0}^{d-1}\sqrt{\lambda_i}|ii\rangle$, where all $\lambda_i\ge 0$ are Schmidt coefficients in decreasing order and $\sum_i \lambda_i=1$. One can define a positive operator $Q$ as
\begin{equation}\label{eq:Q}
	Q=\sum_{\substack{i,j=0,\\i<j}}^{d-1} \sqrt{\lambda_i\lambda_j}(|ij\rangle-|ji\rangle)(\langle ij|-\langle ji|),
\end{equation}
which can be used for constructing an EW for $|\phi\rangle$.
\begin{lemma}
	The partial transpose of $Q$ provides an optimal EW $\mathcal{W}_o^{|\phi\rangle}$, which $\mathcal{W}_o^{|\phi\rangle}$ reads
	\begin{equation}
		\mathcal{W}_o^{|\phi\rangle}=Q^{\Gamma}=\sum_{i,j=0}^{d-1} \sqrt{\lambda_i\lambda_j}|ij\rangle\langle ij|-|\phi\rangle\langle\phi|
	\end{equation}
\end{lemma}
\begin{proof}
To prove that the $\mathcal{W}_o^{|\phi\rangle}$ is an EW, note that it is of the form $Q^{\Gamma}$ with $Q$ being positive semidefinite ($Q\succeq 0$). Thus for all separable states $Tr(\rho_{sep}\mathcal{W}_o^{|\phi\rangle}) = Tr(\rho_{sep}^{\Gamma}Q)\ge 0$. Meanwhile, $Tr(\mathcal{W}_o^{|\phi\rangle}|\phi\rangle\langle\phi|)=\sum_{i}\lambda_i^2-1=-\sum_{i\ne j}\lambda_i\lambda_j<0$. Then $\mathcal{W}_o^{|\phi\rangle}$ is an EW by definition.

To show the optimality of $\mathcal{W}_o^{|\phi\rangle}$, it is sufficient to prove that the set of pure separable states $\{|\phi_1\rangle\otimes|\phi_2\rangle\}$ satisfying $\langle\phi_1|\langle\phi_2|\mathcal{W}_o|\phi_1\rangle|\phi_2\rangle=0$ span the whole Hilbert space $\mathcal{H}_d\otimes\mathcal{H}_d$ \cite{PhysRevA.62.052310}. For qubit case, one has $\mathcal{W}_o^{(2)}=\sqrt{\lambda_0\lambda_1}(|01\rangle-|10\rangle)(\langle 01|-\langle10|)^{\Gamma}$. It is easy to verify that the set of separable states $\{|00\rangle,~(|0\rangle+|1\rangle)(|0\rangle+|1\rangle)/2,~(|0\rangle+i|1\rangle)(|0\rangle-i|1\rangle)/2,~|11\rangle\}$ satisfying $Tr(\rho_{sep}\mathcal{W}_o^{(2)}) = 0$. This set of states span the whole Hilbert space $\mathcal{H}_2\otimes\mathcal{H}_2$. In fact, it has been shown that any decomposable EW acting on $\mathcal{H}_2\otimes\mathcal{H}_d$ is optimal iff it takes the form $\mathcal{W}=Q^{\Gamma}$ for some $Q\succeq 0$ \cite{Augusiak_2011}.

Similarly, in the qudit case ($d>2$), there exist separable states $\{|ee\rangle,~(|e\rangle+|f\rangle)(|e\rangle+|f\rangle)/2,~(|e\rangle+i|f\rangle)(|e\rangle-i|f\rangle)/2,~|ff\rangle\}$ satisfying $Tr(\rho_{sep}\mathcal{W}_o^{|\phi\rangle})=0$, for each pair $0\le e< f \le d-1$. These states span the same space with $\{|ee\rangle,~|ef\rangle,~|fe\rangle,~|ff\rangle\}$. By iterating over all $e<f$, one ends up with a set of separable states spanning the whole space $\mathcal{H}_d\otimes\mathcal{H}_d$. Thus the EW $\mathcal{W}_o^{|\phi\rangle}$ is optimal. This finishes the proof.
\end{proof}

\subsection{Detection of bipartite unfaithful state}

Remarkably, for the $|\phi\rangle$, the most widely used fidelity witness reads $\mathcal{W}_F^{|\phi\rangle} = \lambda_0 I -|\phi\rangle\langle\phi|$. It is straightforward to observe that $\mathcal{W}_F^{|\phi\rangle} -\mathcal{W}_o^{|\phi\rangle}\succeq 0$, which means that the $\mathcal{W}_o^{|\phi\rangle}$ is finer than the $\mathcal{W}_F^{|\phi\rangle}$. This leads to a byproduct that the $\mathcal{W}_o^{|\phi\rangle}$ can detect unfaithful states. Unfaithful states are entangled states which can not be detected by all fidelity witnesses \cite{PhysRevLett.124.200502}, namely, an entangled state $\rho$ is unfaithful if and only if $Tr(\rho W_F^{|\psi\rangle}) \ge 0$ for all $|\psi\rangle$. Therefore, the relationship $\mathcal{W}_F^{|\phi\rangle} -\mathcal{W}_o^{|\phi\rangle} \succeq 0$ itself is not sufficient to demonstrate that the extra entangled states detected by $\mathcal{W}_o^{|\phi\rangle}$ is unfaithful. And a further clarification is required to justify the statement that $\mathcal{W}_o^{|\phi\rangle}$ detects unfaithful state.

Now we would like to provide qualitative and quantitative characterization on the ability to detect unfaithfulness of the $\mathcal{W}_{o}^{|\phi\rangle}$. Consider the class of states $\rho_{|\phi\rangle}(p)=pI/d^2+(1-p)|\phi\rangle\langle\phi|$. From the Observation 1 in Ref.\ \cite{PhysRevLett.126.140503}, it is known that such states are faithful if and only if it is detected by $\mathcal{W}_m=I/d-|\phi_d^+\rangle\langle\phi_d^+|$, with $|\phi_d^+\rangle$ being the maximally entangled state $\sum_{i=0}^{d-1}1/\sqrt{d}|ii\rangle$. By solving $Tr(\rho_{|\phi\rangle}(p)\mathcal{W}_{m})=0$, we obtain that the white noise tolerance of $\mathcal{W}_m$ for $|\phi\rangle$ is
\begin{equation}
	p_f^{|\phi\rangle}=\frac{\sum_{i,j=0}^{d-1}\sqrt{\lambda_i\lambda_j}-1}{\sum_{i,j=0}^{d-1}\sqrt{\lambda_i\lambda_j}-\frac{1}{d}}.
\end{equation}
That is, $\rho_{|\phi\rangle}(p)$ is faithful when $p<p_f^{|\phi\rangle}$.
Similarly, one can obtain the white noise tolerance of $\mathcal{W}_o^{|\phi\rangle}$ for $|\phi\rangle$, which is
\begin{equation}
	p_o^{|\phi\rangle}=\frac{1-\sum_{i=0}^{d-1}\lambda_i^2}{1-\sum_{i=0}^{d-1}\lambda_i^2+\frac{1}{d^2}(\sum_{i,j=0}^{d-1}\sqrt{\lambda_i\lambda_j}-1)}.
\end{equation}
It can be observed that
\begin{equation}
	\begin{aligned}
		\frac{1/p_{o}^{|\phi\rangle}-1}{1/p_{f}^{|\phi\rangle}-1}=&\frac{(\sum_{i,j=0}^{d-1}\sqrt{\lambda_i\lambda_j}-1)^2}{d(d-1)(1-\sum_i \lambda_i^2)} \\
		=&1+\frac{(\sum_{i\ne j}\sqrt{\lambda_i\lambda_j})^2-d(d-1)(1-\sum_i \lambda_i^2)}{d(d-1)(1-\sum_i \lambda_i^2)} \\
  		\le&1+\frac{d(d-1)(\sum_{i\ne j}\lambda_i\lambda_j)-d(d-1)(1-\sum_i \lambda_i^2)}{d(d-1)(1-\sum_i \lambda_i^2)} \\
		=&1+\frac{d(d-1)((\sum_{i=0}^{d-1}\lambda_i)^2-1)}{d(d-1)(1-\sum_i \lambda_i^2)}=1.
	\end{aligned}
\end{equation}
In other words, $p_{f}^{|\phi\rangle} \le p_{o}^{|\phi\rangle}$, where the inequality comes from the Cauchy–Schwarz inequality, and takes equality if $d=2$ or $\lambda_i=1/d$ for all $i$. Therefore, the $\mathcal{W}_o^{|\phi\rangle}$ can always detect unfaithful state $\rho_{|\phi\rangle}(p)$ for $p\in[p_f^{|\phi\rangle},~p_o^{|\phi\rangle})$, unless $|\phi\rangle$ being a two-qubit state or maximally entangled.

\begin{table}[t]
	\centering
	\begin{tabular}{|c |c |c |c |c |c |}
		\hline
		d &3 &4&5&6&7 \\
		\hline
		$l_d$&0.2679&0.4202&0.5195&0.5896&0.6624\\
		\hline
	\end{tabular}
	\caption{
	Maximal unfaithful length $l_d$ from the class of EWs $\mathcal{W}_o^{|\phi\rangle}$. We remark that the optimization of $l_d$ may arrive at a local maximum. We use enough random starting points to support the claim that we arrive at the global maximum.
	}
	\label{tab:unfaithful_max}
\end{table}

As a quantitative investigation, we numerically maximize the interval length $l_d$ of $[p_f^{|\phi\rangle},~p_o^{|\phi\rangle})$ over all $|\phi\rangle$ for different local dimension $d$. We name $l_d$ the maximal unfaithful length from the class of EWs $\mathcal{W}_o^{|\phi\rangle}$, and the results are listed in Table. \ref{tab:unfaithful_max} for $d=3,4,\cdots,7$. It can be seen that $l_d$ grows significantly with an increasing dimension $d$, indicating that the $\mathcal{W}_o^{|\phi\rangle}$ can greatly outperform the fidelity witness. This is also in agreement with the statement that most states are unfaithful as claimed in Ref.\ \cite{PhysRevLett.124.200502}. 

Except for the $l_d$, one may be also interested in the average performance of this new class of EWs on unfaithfulness detection. As a comparison, it is natural to consider two interval $[p_e^{|\phi\rangle}, p_f^{|\phi\rangle})$ and $[p_o^{|\phi\rangle}, p_f^{|\phi\rangle})$, where the $p_e^{|\phi\rangle}$ is the critical value such that $\rho_{|\phi\rangle}(p)$ becomes separable. The former interval contains all unfaithful $\rho_{|\phi\rangle}(p)$, while the latter contains the part that can be detected by the class of $\mathcal{W}_o^{|\phi\rangle}$. Then one can use $avg_{|\phi\rangle} (p_o^{|\phi\rangle}-p_f^{|\phi\rangle})/(p_e^{|\phi\rangle}-p_f^{|\phi\rangle})$ to evaluate the average performance of $\mathcal{W}_o^{|\phi\rangle}$ for detecting unfaithfulness, as shown in Table. \ref{tab:unfaithful_avg}. It is observed that a large percentage of unfaithful states have been detected. This is also the premise that GME witnesses constructed from this class of bipartite EWs can detect multipartite unfaithful state.

\begin{table}[hbtp]
	\centering
	\begin{tabular}{|c |c |c |c |c |c |}
		\hline
		d &3 &4&5&6&7 \\
		\hline
		$avg_{|\phi\rangle} (p_o^{|\phi\rangle}-p_f^{|\phi\rangle})$ &0.0804&0.0969&0.0963&0.0909&0.0848\\
		\hline
		$avg_{|\phi\rangle} (p_e^{|\phi\rangle}-p_f^{|\phi\rangle})$ &0.1190&0.1460&0.1457&0.1379&0.1286\\
		\hline 
		$avg_{|\phi\rangle} \frac{p_o^{|\phi\rangle}-p_f^{|\phi\rangle}}{p_e^{|\phi\rangle}-p_f^{|\phi\rangle}}$ &0.5605&0.5937&0.6089&0.6181&0.6248\\
		\hline
	\end{tabular}
	\caption{
	Average performance of $\mathcal{W}_o^{|\phi\rangle}$ for detecting unfaithfulness. For different local dimension $d$, the average is taken by randomly generating $10^7$ pure states in $\mathcal{H}_d \otimes \mathcal{H}_d$. Since any pure bipartite state admits a Schmidt decomposition $|\phi\rangle = \sum_i \sqrt{\lambda_i} |ii\rangle$, we replace the randomly generated pure bipartite states with random vectors $(\sqrt{\lambda_0},\cdots,\sqrt{\lambda_{d-1}})$ uniformly distributed on the $d$-dimensional unit sphere. Moreover, the critical value $p_e^{|\phi\rangle}$ is $\frac{d^2\sqrt{\lambda_0\lambda_1}}{1+d^2\sqrt{\lambda_0\lambda_1}}$ according to the results in Ref.\ \cite{PhysRevA.59.141}, assuming that the Schmidt coefficients $\lambda_i$ are in decreasing order.}
	\label{tab:unfaithful_avg}
\end{table}

\subsection{Generalization of Lemma 1}\label{sec:appendix generalization}

Finally, we provide a generalization of Lemma 1. For the above entangled state$|\phi\rangle$, one can construct another positive operator
\begin{equation}\label{eq:general_Q}
	\tilde{Q}=\sum_{\substack{i,j=0,\\i<j}}^{d-1} (\alpha_{ij}|ij\rangle-\beta_{ij}|ji\rangle)(\alpha_{ij}\langle ij|-\beta_{ij}\langle ji|),
\end{equation}
instead of $Q$, where $\alpha_{ij}\beta_{ij}=\sqrt{\lambda_i\lambda_j}$ and the $\alpha_{ij},\beta_{ij}$ are all positive. The operator $\tilde{\mathcal{W}}_o=\tilde{Q}^{\Gamma}$ is also optimal EW and applicable in our framework for GME witness construction. Here, the proof of the optimality of $\tilde{\mathcal{W}}_o$ is similar to the case in Lemma 1. It is sufficient to verify that the set of state $\{|ee\rangle,~|ff\rangle,~(\sqrt{\beta_{ef}}|e\rangle + \sqrt{\alpha_{ef}}|f\rangle)\otimes(\sqrt{\alpha_{ef}}|e\rangle + \sqrt{\beta_{ef}}|f\rangle),~(\sqrt{\beta_{ef}}|e\rangle + i\sqrt{\alpha_{ef}}|f\rangle)\otimes(\sqrt{\alpha_{ef}}|e\rangle - i\sqrt{\beta_{ef}}|f\rangle)\}_{e,f=0}^{d-1}$ have zero expectation value when measured with $\tilde{\mathcal{W}}_o$, and span the whole $d^2$-dimensional Hilbert space.

\section{Proof of the examples} 

In this section, we will show explicitly how this construction can be applied to some commonly used multipartite entangled states, and make further discussions on the results.

\subsection{$W$-state}\label{sec:appendix w state}

The $W$-state is an important class of multiqubit entangled states. A class of EWs for $W$-state which can outperform significantly than the fidelity witness has been proposed in Ref.\ \cite{Bergmann_2013}. In Ref.\ \cite{Bergmann_2013}, the authors construct an operator at first, and then prove that this operator is decomposable bipartite EW with respect to all possible bipartitions. While our construction is in the opposite direction. We construct a complete set of bipartite EWs for $W$-state, and lift them to a GME witness. Although different method has been used, our construction recovers the result in Ref.\ \cite{Bergmann_2013}.

We start with the simplest 3-qubit case, where the target state is 
\begin{equation}
      |W_3\rangle=\frac{1}{\sqrt{3}}(|001\rangle+|010\rangle+|100\rangle).
    \end{equation}
For the bipartition $1|23$, the EW constructed from Lemma 1 is of the form
\begin{equation}
	\begin{aligned}
    	\mathcal{W}_{1|23}^{|W_3\rangle}
    	&=\frac{\sqrt{2}}{3}(|000\rangle\langle000|+|1\psi^+\rangle\langle1\psi^+|) \\
    	&~~~~+\frac{2}{3}|0\psi^+\rangle\langle0\psi^+|+\frac{1}{3}|100\rangle\langle100|-|W_3\rangle\langle W_3|,
    \end{aligned}
\end{equation}
with $|\psi^+\rangle=(|01+10\rangle)/\sqrt{2}$. And for the other two bipartitions, the $\mathcal{W}_{2|13}$ and $\mathcal{W}_{3|12}$ can be obtained after permutation of qubits. Then for $|W_3\rangle$, the set $\mathcal{S}$ reads
\begin{equation}
	\begin{aligned}
    	\mathcal{S}=\{&|000\rangle\langle000|,|001\rangle\langle001|,~|010\rangle\langle010|,~|100\rangle\langle100|,\\
        &|0\psi^+\rangle\langle0\psi^+|,~|0_2\psi_{13}^+\rangle\langle0_2\psi_{13}^+|,~|\psi^+0\rangle\langle\psi^+0|,\\
        &|1\psi^+\rangle\langle1\psi^+|,~|1_2\psi_{13}^+\rangle\langle1_2\psi_{13}^+|,~|\psi^+1\rangle\langle\psi^+1|\}.
    \end{aligned}
\end{equation}
These states in $\mathcal{S}$ can be grouped into 3 subsets according to the procedure in the main text:
\begin{equation}
	\begin{aligned}
    	\mathcal{S}_1&=\{|000\rangle\langle000|\},\\
    	\mathcal{S}_2&=\{|001\rangle\langle001|,~|010\rangle\langle010|,~|100\rangle\langle100|,\\
    	&~~~~~~~|0\psi^+\rangle\langle0\psi^+|,~|0_2\psi_{13}^+\rangle\langle0_2\psi_{13}^+|,~|\psi^+0\rangle\langle\psi^+0|\}\\
    	\mathcal{S}_3&=\{|1\psi^+\rangle\langle1\psi^+|,~|1_2\psi_{13}^+\rangle\langle1_2\psi_{13}^+|,~|\psi^+1\rangle\langle\psi^+1|\},
	\end{aligned}
\end{equation}
and the corresponding $\alpha_k$ by Theorem 1 is
\begin{equation}
    \alpha_1=\sqrt{2}/3, ~\alpha_2=2/3,~\alpha_3=\sqrt{2}/3
\end{equation}
respectively. This result in a GME witness
\begin{equation}
	\begin{aligned}
    	\mathcal{W}_{|W_3\rangle}=&\frac{\sqrt{2}}{3}(|000\rangle\langle000|+|101\rangle\langle101|+|011\rangle\langle011|+|110\rangle\langle110|)\\
    	&+\frac{2}{3}(|001\rangle\langle001|+|010\rangle\langle010|+|100\rangle\langle100|)-|W_3\rangle\langle W_3|.
    \end{aligned}
\end{equation}
Moreover, by employing the generalization of Lemma 1 in Eq.\ (\ref{eq:general_Q}), one obtains
\begin{equation}
    \mathcal{W}'_{1|23}=\left[(a|0\rangle|00\rangle-b|1\rangle|\psi^+\rangle)(a\langle0|\langle00|-b\langle1|\langle\psi^+|)\right]^{\Gamma_1},    
\end{equation}
where $a$, $b$ are positive numbers and satisfy $ab=\sqrt{2}/3$. The other two bipartite EWs are obtained immediately after rearrangement of the qubits. For this set of bipartite EWs, the EW $\mathcal{W}_{|W_3\rangle}$ can be generalized into
\begin{equation}
	\begin{aligned}
    	\mathcal{W}_{|W_3\rangle}'=&a^2|000\rangle\langle000|+b^2(|101\rangle\langle101|+|011\rangle\langle011|+|110\rangle\langle110|)\\
    	&+\frac{2}{3}(|001\rangle\langle001|+|010\rangle\langle010|+|100\rangle\langle100|)-|W_3\rangle\langle W_3|.
    \end{aligned}
\end{equation}

In $n$-qubit cases, if a subsystem $A$ contains $m$ qubits, the corresponding bipartite EW from the Lemma 1 is of the form ($1\le m\le n-1$)
\begin{equation}
      \mathcal{W}_{m|n-m}=\sqrt{\frac{m(n-m)}{n^2}}|\psi\rangle_{m|n-m}\langle\psi|^{\Gamma_A},
\end{equation}
where $|\psi\rangle_{m|n-m}={|0\rangle^{\otimes m}}_A{|0\rangle^{\otimes n-m}}_{\bar{A}}-|W_m\rangle_A|W_{n-m}\rangle_{\bar{A}}$. Then the set $\mathcal{S}$ for $|W_n\rangle$ can still be grouped into 3 subsets:
\begin{equation}
    \{|0^{\otimes n}\rangle\},~\{|\pi_m(0^{\otimes n-1}1)\rangle\},~\{|\pi_m'(0^{\otimes n-2}1^{\otimes 2})\rangle\},
\end{equation}
with the corresponding coefficients $\alpha_k$ being
\begin{equation}
    \begin{aligned}
      \alpha_1&=\max_m \sqrt{\frac{m(n-m)}{n^2}}=\frac{\sqrt{\lfloor n/2\rfloor(n-\lfloor n/2\rfloor)}}{n},\\
      \alpha_2&=\max_m \frac{n-m}{n}=\frac{n-1}{n}, \\
      \alpha_3&=\max_m \sqrt{\frac{m(n-m)}{n^2}}=\frac{\sqrt{\lfloor n/2\rfloor(n-\lfloor n/2\rfloor)}}{n}.
    \end{aligned}
\end{equation}
Therefore we arrive at the following $\mathcal{W}_{|W_n\rangle}$,
\begin{equation} \label{eq:w_state}
	\mathcal{W}_{|W_n\rangle}=\frac{n-1}{n}\mathcal{P}_1+\frac{\sqrt{\lfloor n/2\rfloor(n-\lfloor n/2\rfloor)}}{n}
	(|0\rangle\langle 0|^{\otimes n} +\mathcal{P}_2)-|W_n\rangle\langle W_n|,
\end{equation}
with $\mathcal{P}_i=\sum_m\pi_m(|0\rangle^{\otimes n-i}|1\rangle^{\otimes i})\pi_m(\langle 0|^{\otimes n-i}\langle 1|^{\otimes i})$, where the summation $m$ is over all possible permutation of $|0\rangle^{\otimes n-i}|1\rangle^{\otimes i}$.

The EW $\mathcal{W}_{|W_n\rangle}$ can also be generalized in a similar manner with the $\mathcal{W}_{|W_3\rangle}$, so as to recover the results of Ref.\ \cite{Bergmann_2013}. Although ending up with the same witness operator, our construction provides a different insight on why the $\mathcal{W}_{|W_n\rangle}$ takes such a form.

\subsection{Graph states}\label{sec:appendix graph state} 

Before discussing the construction of GME witnesses for the graph states, we first give a brief introduction to the graph states. A graph is a pair $G=(V,E)$ of sets, where the elements of $V$ are called vertices, and the elements of $E$ are edges connecting the vertices. For example, $(1,2) $ represents the edge connecting vertex $1$ and $2$. Two vertices are called neighboring if they are connected by an edge. A graph can also be represented by the adjacency matrix $\Gamma$ with
\begin{equation}
	\Gamma_{ij} = \left\{
	\begin{aligned}
  		&1,  \quad if ~(v_i,~v_j) \in E,\\
  		&0,  \quad otherwise.
	\end{aligned}
	\right.
\end{equation}
Then, an $n$-qubit graph state $|G\rangle$ is defined with an $n$-vertex graph $G$ whose vertices correspond to qubits and edges correspond to control-Z (C-Z) gate between two qubits. Graph state can be expressed with a set of stabilizers
\begin{equation}
	g_i = X_i\prod_{j \in N(i)} Z_j, i = 1,\cdots,n,
\end{equation}
where $X_i$ and $Z_i$ are Pauli operators on qubit (vertex) $i$, and $N(i)$ is the neighborhood of $i$ (\textit{i.e.} the set of vertices directly connected to $i$ by edges). These operators $g_i$ commute with each other and $|G\rangle$ is the common eigenstate of them such that
\begin{equation}
	\forall i=1,\cdots,n,~g_i|G\rangle=|G\rangle,
\end{equation}
Moreover, all the $2^n$ common eigenstates of these $g_i$ form a basis named graph state basis. Each term in this basis is uniquely decided by the eigenvalues of $g_i$. As the eigenvalues of $g_i$ are either $1$ or $-1$, the graph state basis can be denoted by a vector $\vec{a} \in \{0,1\}^n$ such that
\begin{equation}\label{eq:graph_state_basis}
	\forall i=1,\cdots,n,~g_i|a_1\cdots a_n\rangle_G = (-1)^{a_i}|a_1\cdots a_n\rangle_G.
\end{equation}
And the density matrix of $|\vec{a}\rangle_G$ is
\begin{equation}
	|a_1\cdots a_n\rangle_G\langle a_1\cdots a_n|=\prod_{i=1}^{n}\frac{(-1)^{a_i}g_i +I}{2}.
\end{equation}
Specially, the graph state $|G\rangle$ is denoted as $|0\cdots0\rangle_G$.

Remarkably, by choosing the graph state basis instead of the computational basis, the calculation of GME witness construction can be greatly simplified, without needing to perform the Schmidt decomposition. Firstly, the partial transposition of a graph state is diagonal under the graph state basis, namely, $|\vec{a}_0\rangle_G\langle \vec{a}_0|^{T_A}$ is of the form $\sum_{\vec{a}} c_{\vec{a}} |\vec{a}\rangle_G\langle \vec{a}|$. Meanwhile, the operator $Q$ in Eq.\ (\ref{eq:Q}) can be seen as a linear combination of all the eigenstates with negative eigenvalue of $|G\rangle\langle G|^{T_A}$. Therefore, when Lemma 1 is applied to the graph state $|G\rangle$, the resulting bipartite EW $\mathcal{W}_{o,A|\bar{A}}^{|G\rangle}$ is diagonal in the graph state basis. In this case, the vectors in the set $\mathcal{S}$ can be taken as the base vectors $|\vec{a}\rangle_G$, such that the construction in Theorem 1 is easy to achieve. In the following, we propose an explicit procedure for finding the decomposition of $\mathcal{W}_{o,A|\bar{A}}^{|G\rangle}$ in the graph state basis.

Firstly, for the given bipartition $A|\bar{A}$, the adjacency matrix $\Gamma$ can be decomposed into following blocks
\begin{equation}\label{eq:adj_matrix}
	\left(
	\begin{aligned}
		&G_A~~\Gamma_{A|\bar{A}} \\
		&\Gamma_{A|\bar{A}}~~G_{\bar{A}}
	\end{aligned}
	\right).
\end{equation}
We denote $k=rank(\Gamma_{A|\bar{A}})$ as the rank of the submatrix $\Gamma_{A|\bar{A}}$. It is known that a graph state can be transformed into tensor product of $k$ Bell states across the partitions $A$ and $\bar{A}$, using C-Z gates within each partition and local complementation operations \cite{PhysRevA.69.062311}. Here the local complementation $\tau_a$ on a vertex $a$ is defined as follows: $\tau_a : G \to \tau_a(G)$, such that the edge set $E'$ of the new graph $\tau_a(G)$ is $E'= E\cup E\left(N\left(a\right),N\left(a\right)\right)-E\cap E\left(N\left(a\right),N\left(a\right)\right)$. The local complementation $\tau_a(G)$ can be implemented with the following local unitary operation \cite{PhysRevA.69.062311}:
\begin{equation}
	U_a(G) = (-iX_a)^{1/2}\prod_{b\in N(a)}(iZ_b)^{1/2}.
\end{equation}
After this operation, $|G\rangle$ is turned into $|\tau_a(G)\rangle$ and the stabilizers of $|G\rangle$ transform according to the following equations:
\begin{equation}\label{eq:stabizer_transfer}
	\begin{aligned}
		&U_a(G)g_b^GU_a(G)=g_a^{\tau_a(G)}g_b^{\tau_a(G)},~if ~b\in N(a); \\
		&U_a(G)g_b^GU_a(G)=g_b^{\tau_a(G)}, ~~~~~~~~~if ~b\notin N(a).
	\end{aligned}
\end{equation}
Meanwhile, we remark that a bipartite EW $\mathcal{W}_{A|\bar{A}}$ for $|G\rangle$ has been transformed into another bipartite EW $\mathcal{W}_{A|\bar{A}}'$ for $|G'\rangle$ after some local unitary operation with respect to $A|\bar{A}$. Hence our task for constructing bipartite EW of the initial graph state $|G\rangle$ has been turned into finding a bipartite EW for $|Bell\rangle^{\otimes k}$ by employing Lemma 1, a much easier task compared with the initial one.

Secondly, after reversing the above transformation process from $|G\rangle$ to $|Bell\rangle^{\otimes k}$, the EW for $|Bell\rangle^{\otimes k}$ which is diagonal in Bell state basis will be turned back into a bipartite EW for $|G\rangle$ which is diagonal in graph state basis.

\begin{figure}[h]
	\centering
	\includegraphics[width=14.72cm,height=4.54cm]{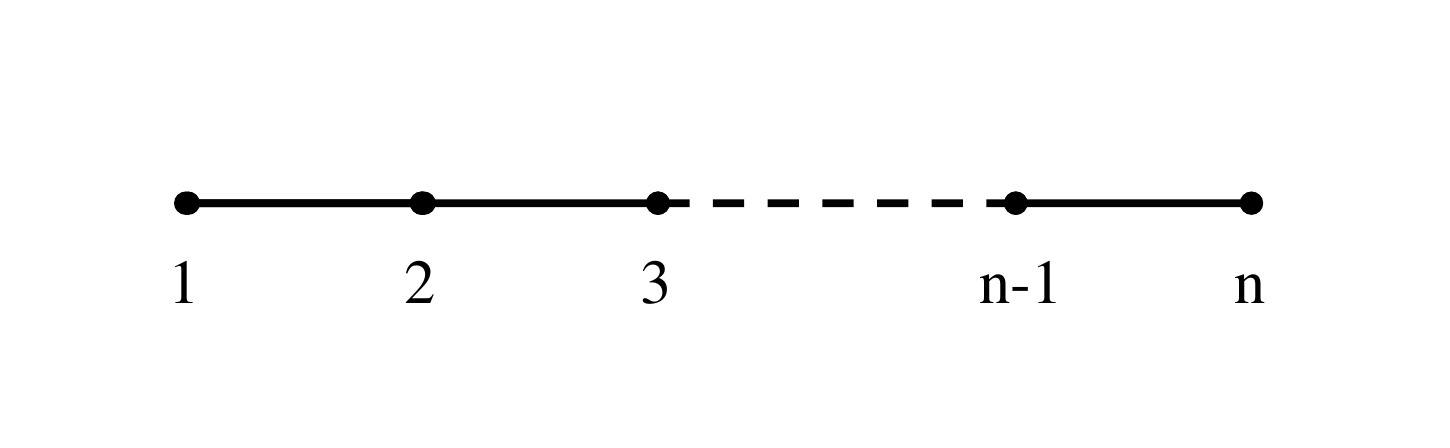}
	\caption{Representation of $n$-qubit linear cluster state with graph, where $n$ qubits are connected one by one with C-Z gates.}
	\label{fig:appendix_fig1}
\end{figure}
With the above foundation, we move on to a explicit discussion on a typical class of graph states: linear cluster state $|Cl_n\rangle$. Linear cluster state is represented with the graph in Fig.\ \ref{fig:appendix_fig1}.

We call the bipartition $A|\bar{A}$ a rank-$k$ bipartition if $rank(\Gamma_{A|\bar{A}})=k$, with the $\Gamma_{A|\bar{A}}$ defined in Eq.\ (\ref{eq:adj_matrix}). All rank-$1$ bipartitions of linear cluster state have only two possible types of the subgraph on the boundary $\Gamma_{A|\bar{A}}$ (Fig.\ \ref{fig:appendix_fig2}).
\begin{figure}[h]
	\centering
	\includegraphics[width=9.11cm,height=5.95cm]{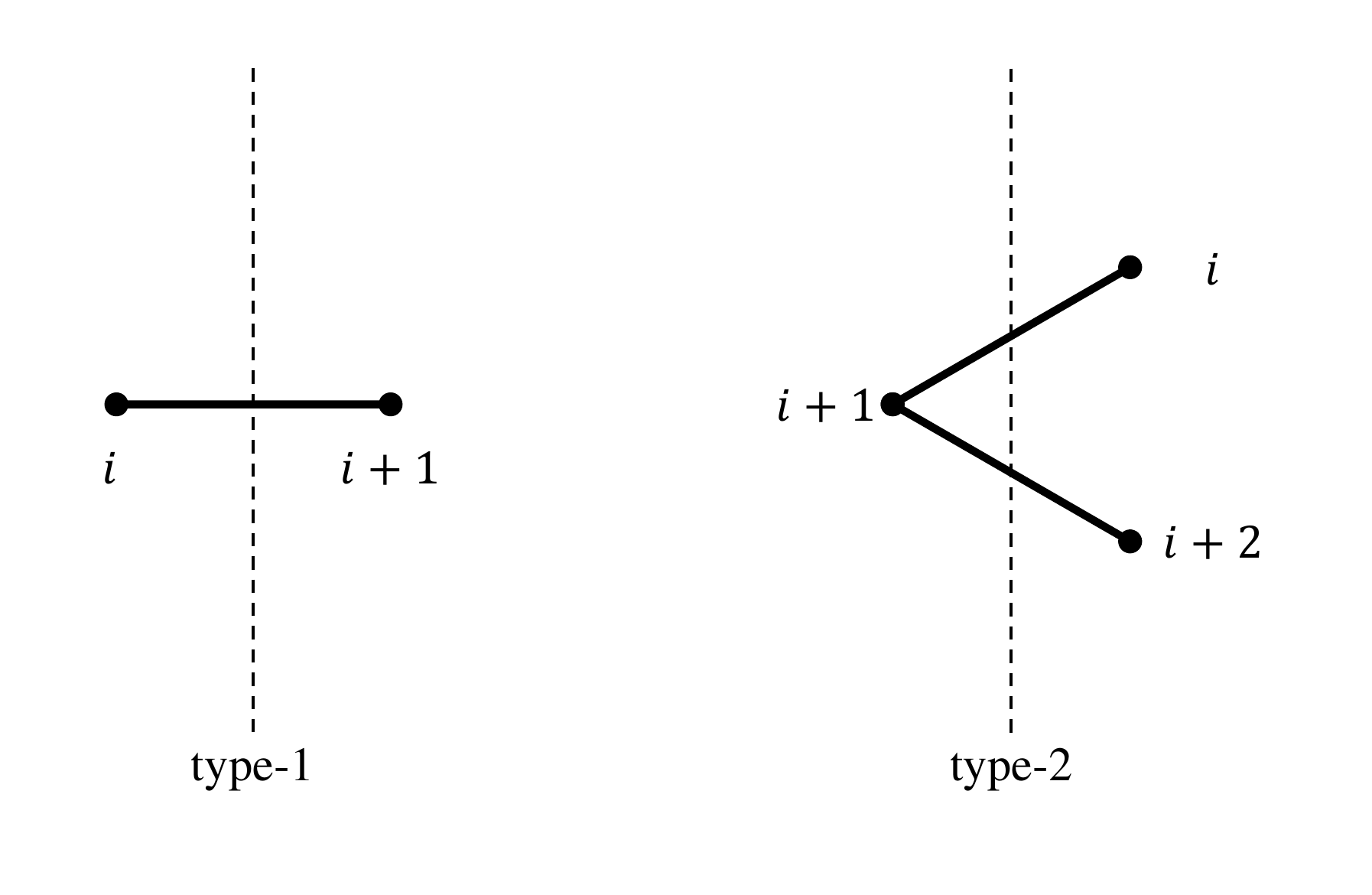}
	\caption{Possible subgraphs on the boundary across rank-1 bipartition. For a given bipartition $A|\bar{A}$, if $rank(\Gamma_{A|\bar{A}}) = 1$, then the subgraph across this bipartition has only the above two types.}
	\label{fig:appendix_fig2}
\end{figure}
Any other edge is deleted by C-Z gates within each partition. For the type-1 subgraph $G_{i,i+1}$, the bipartite EW reads
\begin{equation}
	\begin{aligned}
		\mathcal{W}_{G_{i,i+1}}=&\frac{1}{2}\left(|0_i0_{i+1}\rangle_{G_{i,i+1}}\langle 0_i0_{i+1}|+|0_i1_{i+1}\rangle_{G_{i,i+1}}\langle 0_i1_{i+1}|+|1_i0_{i+1}\rangle_{G_{i,i+1}}\langle 1_i0_{i+1}|\right.  \\
		&\left. +|1_i1_{i+1}\rangle_{G_{i,i+1}}\langle 1_i1_{i+1}|\right)    
		-|G_{i,i+1}\rangle\langle G_{i,i+1}|.
	\end{aligned}
\end{equation}
by employing Lemma 1, where the state vectors like $|0_i1_{i+1}\rangle_{G_{i,i+1}}$ are graph state basis defined in the Eq.\ (\ref{eq:graph_state_basis}), and the `$0$'s on other vertices are omitted for simplicity here and after. Note that the $g_j^{G_{i,i+1}}$ can be transformed back to the $g_j^{Cl_n}$ by employing C-Z gates without disturbing the eigenvalue of the state vector. Therefore the bipartite EW for the original state $|Cl_n\rangle$ is
\begin{equation}
	\begin{aligned}
		\mathcal{W}_{A|\bar{A}}=&\frac{1}{2}\left(|Cl_n\rangle\langle Cl_n|+|0_i1_{i+1}\rangle_{Cl_n}\langle 0_i1_{i+1}|+|1_i0_{i+1}\rangle_{Cl_n}\langle 1_i0_{i+1}|\right. \\
		&\left. +|1_i1_{i+1}\rangle_{Cl_n}\langle 1_i1_{i+1}|\right)-|Cl_n\rangle\langle Cl_n|,
	\end{aligned}
\end{equation}
when formulated in the graph state basis. After normalizing the $\mathcal{W}_{A|\bar{A}}$ to meet the constraint $Tr(\mathcal{W}_{A|\bar{A}}|Cl_n\rangle\langle Cl_n|)=-1$, we obtain
\begin{equation}
	\mathcal{W}_{A|\bar{A}}'=|0_i1_{i+1}\rangle_{Cl_n}\langle 0_i1_{i+1}|+|1_i0_{i+1}\rangle_{Cl_n}\langle 1_i0_{i+1}|+|1_i1_{i+1}\rangle_{Cl_n}\langle 1_i1_{i+1}|-|Cl_n\rangle\langle Cl_n|,
\end{equation}
as the bipartite EW used in our construction. This bipartite EW contributes the following terms to the set $\mathcal{S}$:
\begin{equation}
	\{|0_i1_{i+1}\rangle_{Cl_n},~|1_i0_{i+1}\rangle_{Cl_n},~|1_i1_{i+1}\rangle_{Cl_n}\}.
\end{equation}
This set is denoted as $\mathcal{S}_{i,type-1}=\{01,~10,~11\}_{i,i+1}$ for short.

Meanwhile, the type-2 subgraph $G_{i,i+1,i+2}$ in Fig.\ \ref{fig:appendix_fig2} can be transformed into the type-1 subgraph after applying $C_{Z,(i,i+2)}$, $U_i(G_{i,i+1,i+2})$ and $C_{Z,(i,i+2)}$ sequentially. We remark that the local complementation operation $U_i(G_{i,i+1,i+2})$ may change the corresponding eigenvalue when $g_j^{G_{i,i+1,i+2}}$ turns into $g_j^{G_{i,i+1}}$, which is decided by the Eq.\ (\ref{eq:stabizer_transfer}). Therefore, this kind of bipartitions contribute the following set of states to the set $\mathcal{S}$:
\begin{equation}
	\mathcal{S}_{i,type-2}=\{010,~101,~111\}_{i,i+1,i+2}.
\end{equation}
In summary, all the rank-$1$ bipartitions contribute an operator $R_1$ by our construction. If one denotes the $V_1$ as the set of vectors from $\{0,1\}^n$ such that the maximal distance between the `$1$'s appearing in each vector is smaller than 3, the $R_1$ can be formulated as
\begin{equation}
	R_1=\sum_{\vec{a}\in V_1}|\vec{a}\rangle_{Cl_n}\langle\vec{a}|.
\end{equation}

For rank-$2$ bipartitions, their boundaries are composed of two rank-$1$ boundaries. For example, if there are two type-$1$ parts, the bipartite EW takes the form
\begin{equation}
	\mathcal{W}_{A|\bar{A}}=\frac{1}{4}\sum_{\vec{a}\in\{0,1\}^4}|\vec{a}_{i,i+1,j,j+1}\rangle_{Cl_n}\langle \vec{a}_{i,i+1,j,j+1}|-|Cl_n\rangle\langle Cl_n|.
\end{equation}
After normalization, the bipartite EW reads
\begin{equation}
	\mathcal{W}_{A|\bar{A}}=\frac{1}{3}\sum_{\substack{\vec{a}\in\{0,1\}^4,\\\vec{a}\ne\vec{0}}}|\vec{a}_{i,i+1,j,j+1}\rangle_{Cl_n}\langle \vec{a}_{i,i+1,j,j+1}|-|Cl_n\rangle\langle Cl_n|.
\end{equation}
Such bipartite EW contribute the following new terms to the set $S$:
\begin{equation}
	S_{i,type-1}\otimes S_{j,type-1}=\{0101,~0110,~0111,~1001,~1010,~1011,~1101,~1110,~1111\}_{i,i+1,j,j+1}.
\end{equation}
The contribution of other possibilities can be decided in a similar manner as above for the type-2 subgraph. All these rank-$2$ bipartitions contribute a set $V_2$ to $\mathcal{S}$. Here a vector from $\{0,1\}^n$ belongs to $V_2$ if there exist at most two `$1$'s whose distance is larger than $2$ at the same time in the vector. Finally, all the rank-$2$ bipartitions introduce an operator $R_2$ to our construction, with
\begin{equation}
	R_2=\sum_{\vec{a}\in V_2}\frac{1}{3}|\vec{a}\rangle_{Cl_n}\langle\vec{a}|.
\end{equation}

For a rank-$k$ bipartition, the subgraph on the boundary is nothing but a combination of $k$ rank-$1$ part. After repeating the above process, it is shown that all the rank-$k$ bipartitions contribute the following operator $R_k$:
\begin{equation}
	R_k=\sum_{\vec{a}\in V_k}\frac{1}{2^k-1}|\vec{a}\rangle_{Cl_n}\langle\vec{a}|.
\end{equation}
A vector $\vec{a}$ belongs to $V_k$ if there exist at most $k$ for the number of `$1$'s in $\vec{a}$, such that their distance with each other are larger than $k$ at the same time. It can be observed immediately that $k\le \lceil n/3 \rceil$, indicating that the partition whose rank is higher than $\lceil n/3 \rceil$ gives no extra contribution.

After considering all the bipartitions, we end up with the GME witness $\mathcal{W}_{Cl_n}$ introduced in the main text, namely, 
\begin{equation}
	\mathcal{W}_{Cl_n}=\sum_{k=1}^{\lceil n/3 \rceil} R_k-|Cl_n\rangle\langle Cl_n|.
\end{equation}
As an example, for $4$-qubit cluster state,
\begin{equation}
	\mathcal{W}_{Cl_4}=\sum_{\vec{a}\in V_1}|\vec{a}\rangle_G\langle \vec{a}|+\frac{1}{3}\sum_{\vec{a}\in V_2}|\vec{a}\rangle_G\langle \vec{a}|-|G\rangle\langle G|,
\end{equation}
where $V_1$ is the set $\{0001,~0010,~0011,~0100,~0101,~0110,~0111,~1000,~1010,~1100,~1110\}$, and $V_2$ is the set $\{1001,~1011,~1101,~1111\}$.

Remarkably, in $4$-qubit case, the best known EW is \cite{PhysRevLett.106.190502}
\begin{equation}
	\mathcal{W}_{Cl_4}^{opt}=\sum_{\vec{a}\in V_1}|\vec{a}\rangle_G\langle \vec{a}|-|G\rangle\langle G|.
\end{equation}
It is finer than the $\mathcal{W}_{Cl_4}$ above. That is, while our approach is already quite powerful, there is still room for improvement. In this particular case, the improvement can be achieved by an elaborate choice of the set of bipartite EWs, instead of using Lemma 1 only. If the bipartite EWs for $13|24$
and $14|23$ in the above construction are replaced by
\begin{equation}
	\begin{aligned}
		\mathcal{W}_{13|24}=&|0001\rangle_{Cl_4}\langle0001| + |0100\rangle_{Cl_4}\langle0100|+|0101\rangle_{Cl_4}\langle0101| +|0011\rangle_{Cl_4}\langle0011| \\
		&+|0110\rangle_{Cl_4}\langle0110| +|0111\rangle_{Cl_4}\langle0111| -|Cl_4\rangle\langle Cl_4|,   \\
		\mathcal{W}_{14|23}=&|0001\rangle_{Cl_4}\langle0001| + |0010\rangle_{Cl_4}\langle0010|+|0101\rangle_{Cl_4}\langle0101| +|0011\rangle_{Cl_4}\langle0011|  \\
		&+|0110\rangle_{Cl_4}\langle0110| +|0111\rangle_{Cl_4}\langle0111| -|Cl_4\rangle\langle Cl_4|,
	\end{aligned}
\end{equation}
respectively, one can recover the $\mathcal{W}_{Cl_4}^{opt}$ with Theorem 1. With this example on $4$-qubit cluster state, we highlight that Lemma 1 is just an alternative choice which ends up with robust GME witnesses. Our construction in fact allows a flexible choice on the set of EWs to be lifted to multipartite case, and a suitable choice can further improve its performance. Moreover, it should be remarked that our discussion was based on the partial transposition throughout this paper, to obtain higher noise resistance. If bipartite EWs in the construction are designed by other positive maps (e.g., the Choi's map), different classes of GME witness can be found. This may help to harness the full potential of Theorem 1 in future work.

\subsection{Multipartite states admitting Schmidt decomposition.}\label{sec:appendix GHZ state} 

A special case of multipartite entangled states is the multipartite states admitting Schmidt decomposition. Without loss of generality, we can assume that such states are of the form $|\phi_s\rangle=\sum_{i=0}^{d-1}\sqrt{\lambda_i}|i\rangle^{\otimes n}$ with $\lambda_i\ge 0 $ in decreasing order. Then the Lemma 1 gives a set of bipartite EWs $\mathcal{W}_{A|\bar{A}}^{|\phi_{SD}\rangle}$:
\begin{equation}
    \mathcal{W}_{A|\bar{A}}^{|\phi_{SD}\rangle}=\sum_{i,j=0}^{d-1} \sqrt{\lambda_i\lambda_j}{|i\rangle^{\otimes k}}_A{|j\rangle^{\otimes n-k}}_{\bar{A}}{\langle i|^{\otimes k}}_A {\langle j|^{\otimes n-k}}_{\bar{A}}-|\phi_s\rangle\langle\phi_s|,
\end{equation}
where $k=|A|$ is the number of qudits in subsystem $A$. For these bipartite EWs, the set $S$ is
\begin{equation}
    \{\pi_m(|i\rangle^{\otimes r}|j\rangle^{\otimes n-r})\}_{r,i,j,\pi_m}\cup\{|l\rangle^{\otimes n}\}_{l=0}^{d-1},
\end{equation}
with $r=1,2,\cdots,n-1$, $i,j=0,1,\cdots,d-1$ ($i<j$) and $\pi_m$ being all possible permutations of $|i\rangle^{\otimes r}|j\rangle^{\otimes n-r}$. Note that all state vectors in $S$ are orthogonal with each other, thus our construction ends up with the following multipartite EW
\begin{equation}
	\begin{aligned}
		\mathcal{W}_{|\phi_s\rangle}=&\sum_{\substack{i,j=0,\\i< j}}^{d-1} \sum_{r=1}^{n-1} \sum_m \sqrt{\lambda_i\lambda_j}\pi_m(|i\rangle^{\otimes r}|j\rangle^{\otimes n-r})\pi_m(\langle i^{\otimes r}|\langle j|^{\otimes n-r}) \\
		&+\sum_{i=0}^{d-1}\lambda_i|i\rangle\langle i|^{\otimes n}-|\phi_s\rangle\langle\phi_s|,
	\end{aligned}
\end{equation}
where the summation of $m$ is over all possible permutations $\pi_m(|i\rangle^{\otimes r}|j\rangle^{\otimes n-r})$ of $|i\rangle^{\otimes r}|j\rangle^{\otimes n-r}$.

Moreover, similar to the case of proving the optimality of $\mathcal{W}_o^{|\phi\rangle}$ in the first section, one can verify the optimality of ${W}_{|\phi_s\rangle}$ by checking that all the biseparable states satisfying $Tr(\mathcal{W}_{|\phi_s\rangle}\rho_{bs})=0$ span the whole Hilbert space $\mathcal{H}_d^{\otimes n}$.

\subsection{GME witness for multi-qubit singlet states}\label{sec:appendix singlet} 

Multi-qubit singlet states are of particular experimental interest, while the GME witness for them is less investigated. In this example, it is shown that our framework works well for the multi-qubit singlet states. In the main text, we provide the result for a specific class of four-qubit singlet states. While here we begin with the discussion on general four-qubit singlet states 
\begin{equation}
	|\varphi_4\rangle = a |\psi_{12}^-\rangle\otimes |\psi_{34}^-\rangle + e^{i\theta}b |\psi_{13}^-\rangle\otimes |\psi_{24}^-\rangle, 
\end{equation} 
with the constraint $a^2+b^2+cos(\theta)ab=1$ and $|\psi_{12}^-\rangle$ being the two-qubit singlet state $(|01\rangle-|10\rangle)/\sqrt{2}$ on the first two qubits. By performing our construction procedure for all four-qubit singlet states, it is observed that the set $\mathcal{S}$ is always divided into $5$ subsets and the identity operators on the corresponding subspaces are just $\{\mathcal{P}_i^4\}_{i=0}^4$ (The $\mathcal{P}_i^4$ has been defined below the Eq.\ (\ref{eq:w_state})). More specifically, the resulting witness is
\begin{equation}
	\mathcal{W}_4 = c_2 \mathcal{P}_2^4 + c_1 (\mathcal{P}_1^4 + \mathcal{P}_3^4) + c_0 (\mathcal{P}_0^4  + \mathcal{P}_4^4) - |\varphi_4\rangle\langle\varphi_4|,
\end{equation}
with the coefficients decided by
\begin{equation}
	\begin{aligned}
  		c_2 &= \max\{1-\frac{3}{4}a^2, 1-\frac{3}{4}b^2, \frac{3}{4}(a^2+b^2) -\frac{1}{2}\},  \\
  		c_1 &= \frac{1}{2},  \\ 
  		c_0 &= \max\{\frac{1}{2}-\frac{1}{4}(a^2+b^2),\frac{1}{4}a^2,\frac{1}{4}b^2\}. 
  \end{aligned}
\end{equation}
Specially, with a choice of $\theta = \pi/2$, this recovers the EW in the main text. While if $a=-1$, $b=1$ and $\theta=0$, $|\varphi_4\rangle$ becomes a biseparable state $|\psi_{14}^-\rangle\otimes |\psi_{23}^-\rangle$ and the corresponding EW become positive semidefinite. 

When the number of qubit grows, achieving a generic expression becomes more complicated. To investigate the GME witness construction in this case, we consider the following six-qubit singlet state 
\begin{equation}
	\begin{aligned}
		|\varphi_6\rangle = &\frac{1}{2}\left(|\psi_{12}^-\rangle\otimes |\psi_{34}^-\rangle \otimes |\psi_{56}^-\rangle + i |\psi_{13}^-\rangle\otimes |\psi_{24}^-\rangle \otimes |\psi_{56}^-\rangle \right. \\
		& \left. + i |\psi_{12}^-\rangle\otimes |\psi_{35}^-\rangle \otimes |\psi_{46}^-\rangle - |\psi_{13}^-\rangle\otimes |\psi_{25}^-\rangle \otimes |\psi_{46}^-\rangle\right),
	\end{aligned}
\end{equation}
for which we arrive at the GME witness 
\begin{equation}
	\mathcal{W}_6= \frac{5}{8} \mathcal{P}_3^6 + \frac{1}{2} (\mathcal{P}_2^6 + \mathcal{P}_4^6) + \frac{1}{4} (\mathcal{P}_1^6 + \mathcal{P}_5^6) + \frac{1}{8} (\mathcal{P}_0^6 + \mathcal{P}_6^6) - |\varphi_6\rangle\langle\varphi_6|.
\end{equation}
Based on these results, it is reasonable to conjecture that for some $2n$-qubit singlet state $|\varphi_{2n}\rangle$, there exists a GME witness taking the form 
\begin{equation}
	\mathcal{W}_{2n}= c_n \mathcal{P}_n^{2n} + \sum_{i=0}^{n-1} c_i(\mathcal{P}_i^{2n} + \mathcal{P}_{2n-i}^{2n})- |\varphi_{2n}\rangle\langle\varphi_{2n}|.
\end{equation}
with $c_i \ge c_{i-1} \ge 0$ for $i=1,\cdots, n$ and $c_n$ is the maximal squared overlap between $|\varphi_6\rangle$ and biseparable states. Moreover, if $c_i$ scales with $(1/2)^{-(n-i+1)}$ as in the four- and six-qubit case, a high white noise tolerance tending to $1$ can be expected for a large number of qubit.

\end{document}